\documentclass[a4paper,12pt]{article}

\usepackage[english]{babel}
\usepackage{multicol, authblk}
\usepackage{color}
\usepackage[dvipsnames,table]{xcolor}
\usepackage{amssymb, dsfont, amsmath, graphicx, amsfonts, latexsym, enumerate, color, amsthm, bm,mathrsfs}
\usepackage{fancyvrb, numprint, lipsum}
\usepackage[utf8]{inputenc}
\usepackage{hyperref, fullpage}
\usepackage{amsxtra,amstext}
\usepackage{setspace}

\usepackage{algorithm}
\usepackage{algpseudocode}
\usepackage{algorithmicx}
\usepackage{soul}
\usepackage{multirow}

\newcommand{\MAP}{M\!\!A\!P}
\newcommand{\theo}{\textrm{latent}}
\newcommand{\nbInd}{N}
\newcommand{\Lik}{\mathcal{L}}
\newcommand{\ind}[1]{1\!\textrm{l}_{#1}}

\newcommand{\ol}[1]{\overline{#1}}

\newcommand{\dis}{\displaystyle}

\newcommand{\e}{{\rm I\!E}}
\newcommand{\var}{{\rm var}}

\newcommand{\p}{{\rm I\!P}}

\DeclareMathOperator*{\argmax}{arg\,max}
\DeclareMathOperator*{\argmin}{arg\,min}
\newcommand\numberSampleVarThe{4}

\newtheorem{thme}{Theorem}

\newtheorem{prop}[thme]{Proposition}

\newtheorem{lem}[thme]{Lemma}

\providecommand{\keywords}[1]
{
  \small
  \textbf{\textit{Keywords: }} #1
}

\definecolor{color1}{rgb}{.2,.6,1}      
\definecolor{noir}{rgb}{0,0,0}
\definecolor{Gris}{gray}{.5}
\definecolor{color2}{rgb}{.6,0.9,0.75}   
\definecolor{color3}{rgb}{.6,0.4,1}      
\definecolor{color4}{rgb}{.35,.80,.70}  
\definecolor{color5}{rgb}{.35,.75,.2}    
\definecolor{red2}{rgb}{1,0.3,0.1}    
\definecolor{mag}{rgb}{0.66,0,0}     
\definecolor{colorfond}{rgb}{.50,.95,.85}
\definecolor{BleuCiel}{rgb}{.2,.5,1}   
\definecolor{Bordeau}{rgb}{.7,0,0}    
\definecolor{color8}{rgb}{.60,.10,.40}   
\definecolor{VertFonce}{rgb}{0,0.6,0}
\definecolor{MidGreen}{rgb}{0.6,1,0.6}
\definecolor{LightGreen}{rgb}{0.88,1,0.88}
\definecolor{LightGray}{rgb}{0.94,0.94,0.94}
\definecolor{VeryLightBlue}{rgb}{0.9,0.9,1}
\definecolor{LightBlue}{rgb}{0.8,0.8,1}
\definecolor{DarkBlue}{rgb}{0,0,0.6}
\definecolor{VeryLightYellow}{rgb}{1,1,0.9}
\definecolor{LightYellow}{rgb}{1,1,0.6}
\definecolor{MidYellow}{rgb}{1,1,0.5}
\definecolor{VeryLightRed}{rgb}{1,0.9,0.9}
\definecolor{LightRed}{rgb}{1,0.8,0.8}
\definecolor{turquoise}{rgb}{0.00,0.53,0.68}{}
\definecolor{turqFonce}{rgb}{0.00,0.5,0.45}{}
\definecolor{mauve}{rgb}{0.50,0.00,0.50}

\definecolor{Soutput}{rgb}{0,0,0.56}
\definecolor{Sinput}{rgb}{0.56,0,0}
\DefineVerbatimEnvironment{Sinput}{Verbatim}
{formatcom={\color{Sinput}},fontsize=\normalsize, baselinestretch=0.9}
\DefineVerbatimEnvironment{Soutput}{Verbatim}
{formatcom={\color{Soutput}},fontsize=\normalsize, baselinestretch=0.9}

\date{}

\title{Reconciling Binary Replicates: Beyond the Average}

\author[1]{M. Royer-Carenzi}
\author[1]{H. Lorenzo}
\author[1]{P. Pudlo}

\affil[1]{Aix Marseille Univ, UMR 7373, CNRS, Centrale Marseille, I2M, Marseille, France}

\begin{document}

\maketitle

\begin{abstract}
Binary observations are often repeated to improve data quality, creating technical replicates. Several scoring methods are commonly used to infer the actual individual state and obtain a probability for each state. The common practice of averaging replicates has limitations, and alternative methods for scoring and classifying individuals are proposed. Additionally, an indecisive response might be wiser than classifying all individuals based on their replicates in the medical context, where 1 indicates a particular health condition. Building on the inherent limitations of the averaging approach, three alternative methods are examined: the median, maximum penalized likelihood estimation, and a Bayesian algorithm. The theoretical analysis suggests that the proposed alternatives outperform the averaging approach, especially the Bayesian method, which incorporates uncertainty and provides credible intervals. Simulations and real-world medical datasets are used to demonstrate the practical implications of these methods for improving diagnostic accuracy and disease prevalence estimation.
\end{abstract}


\section{Introduction}

In Medicine, Biology, and, more generally, applied sciences, conclusions are drawn from carefully analyzing noisy data. Noise may come from the measurement protocol, the instruments, or the intrinsic variability of the phenomenon under study. In this context, using technical replicates is a common practice to improve the quality of the conclusions. Approaches based on replication groups have been studied to reduce the false positive rate in replicate analysis, see for example~\cite{TCDZ10}. As Palmer said in \cite{Palmer1986}, ``Measurement errors are unavoidable, and repetitions should be the rule to quantify its magnitude''. However, dealing with technical replicates is not always straightforward: replicates of the same individuals are more likely to be similar than replicates of different individuals. This phenomenon is known as overdispersion, and it has been widely studied in the literature~\cite{Will82, Jae08, Im21, HML90}. Jaeger~\cite{Jae08} showed that not taking into account the correlation of the responses for the same subject leads to biased estimates and artificially increases our trust in the estimators, which can lead to invalid statistical analyses. A common practice to deal with this correlation is to summarize the replicates of the same individual by a single score, which is usually the empirical average of the replicates. The distinction between technical and biological replicates is essential to avoid misinterpretation in medical studies~\cite{TKKSG19,VFC12}.
This paper focuses on binary data: each replicate is a binary variable, and the individual's actual state is also a binary variable. The leading difficulty is that a binary variable carries very little information, and noise can change its value to its opposite. In medical diagnostics, the actual state equals $1$ when the subject has the disease and $0$ otherwise. The replicates are binary values that aim to infer the individual's actual state.

Dealing with binary replicates in medical studies involves addressing the challenges of repeated binary outcomes, which are common in longitudinal studies and clinical trials \cite{WFT01, SSAR18, RKPSC20, LFSLS22}. These outcomes are often correlated, necessitating specialized statistical methods for accurate analysis. The use of correlated binomial models has been proposed to take into account the dependence between technical replicates, enabling better estimation of error rates~\cite{zhang19}. We need at least to account for within-patient correlation in these data. Various approaches have been developed to handle these complexities, each with advantages and limitations. The recent literature dates back to 1982~\cite{Will82}. In 1990, Hujoel et al.~\cite{HML90} introduced a correlated binomial model that can be employed to estimate the sensitivity and specificity of these diagnostic tests; see also \cite{ahn1997statistical,ahn1997evaluation}. Various approaches have been developed to construct robust confidence intervals in the presence of correlated binary data~\cite{SCWLM20}. Contemporary methods emphasize integrating advanced biomarkers and machine learning algorithms for improved diagnostic performance. Yet there has been recent interest in improving the results on repeated binary data, with modern methods such as the Bayesian machinery of a latent class model~\cite{keddie2023estimating}. Two-phase sampling has been used to improve the efficiency of parameter estimation in longitudinal binary data contexts~\cite{TMHMRHS21}. However, few theoretical results have been published on the accuracy of these methods. This paper intends to fill in this gap in a relatively simple case.

We introduce and compare four statistical methods: an average-based, a median-based, a maximum-a-posteriori-based, and a Bayesian method. Each method provides a different way of scoring the individuals. We design a classifier for each approach that makes decisions for each individual and returns a trustworthy diagnostic. We propose different methods to estimate the prevalence of the disease in the sampled population and the sensitivity and specificity of the binary replicate without resorting to a gold standard test. With a relatively simple framework, we intend to provide theoretical results that compare the proposed methods illustrated in various numerical cases. Moreover, we introduce an indecision outcome in the individuals' classification, which is returned in cases of significant doubt. This helps us understand better the information provided by the multiple scorings of the individuals and better fit real-world medical instances in which erroneous decisions can have serious consequences.

We organized the paper as follows. Section~\ref{sec:methods} describes the statistical model, the four competing approaches, and their rationale. Section~\ref{sec:theoretical} gives mathematical results to compare their efficiency in terms of both classification and estimation. In Section~\ref{Sec:NumRes}, simulations are used to compare the methods, and numerical illustrations on medical datasets are provided: a periodontal dataset \cite{HML90} and a synthetic mammogram screening dataset \cite{BCS03, KiLe17}. The methods described in the paper are implemented in an R-package, available on GitHub at the following address: \href{https://github.com/pierrepudlo/BinaryReplicates}{https://github.com/pierrepudlo/BinaryReplicates}.

\section{Methods}\label{sec:methods}
In Section~\ref{sec:model}, we introduce a statistical model for binary replicates that intends to capture the binary state of an individual. To infer the prevalence of each state in the population and to classify the individuals, we introduce four statistical methods: an average-based, a median-based, a maximum-a-posteriori-based, and a Bayesian method. 
All these methods rely on scoring the individuals given their binary replicates; see Sections~\ref{sec:scoring_AM}, \ref{sec:EM}, and \ref{sec:Bayesian_scoring}. Finally, for each method, we introduce classifiers to recover the binary state of each individual based on the scoring and estimators of the prevalence in Section~\ref{sec:scoring_to_classifier}.

\subsection{The statistical model}\label{sec:model}

We consider a dataset comprising binary observations $X_{ij} \in {0, 1}$, where $i = 1, \ldots, {\nbInd}$ indexes individuals and $j = 1, \ldots, n_i$ indexes technical replicates. Each replicate $X_{ij}$ approximates the true status $T_i \in {0, 1}$, used to mitigate measurement imprecision. For instance, in medical diagnostics, $T_i = 1$ denotes disease presence, while $T_i = 0$ indicates health. We assume the replicates $X_{ij}$ are independent noisy measurements of $T_i$.
The marginal distribution of $T_i$ is Bernoulli with parameter $\theta_T$, representing disease prevalence. For each replicate $j$,
\begin{align*}
    p & = \p(X_{ij}=1 | T_i=0) \quad \text{is the false-positivity rate,} \\
    q & = \p(X_{ij}=0 | T_i=1) \quad \text{is the false-negativity rate.}
\end{align*}
The measurement error $T_i - X_{ij}$ takes values in ${-1, 0, 1}$, with $0$ indicating no error, $1$ a false positive, and $-1$ a false negative.

By the law of total probability, $\mathbb{E}(T_i - X_{ij}) = \theta_T q + (\theta_T - 1)p$, which is non-zero if $\theta_T \neq p / (p + q)$. Thus, the diagnostic test is biased under these conditions. If either $p$ or $q$ exceeds $1/2$, the replicate is more likely wrong than correct, resulting in unreliable data. We assume $p, q \in (0, 1/2)$. In what follows, we assume that $p,q\in(0, 1/2)$.

\subsection{Sufficient statistics} \label{sec:sufficient}
We summarize the entire dataset of $X_{ij}$, for $i = 1, \ldots, {\nbInd}$ and $j = 1, \ldots, n_i$, by introducing the sum $S_i = X_{i1} + \cdots + X_{in_i}$ for each individual $i$. This sum represents the number of positive replicates for individual $i$. The vector $(S_1, \ldots, S_{\nbInd})$ is a sufficient statistic for the dataset. Indeed, the only information lost in this transformation is the order of the replicates for each individual. However, since the technical replicates for a given individual are exchangeable, their order carries no relevance to the statistical analysis. Therefore, no statistical information is lost when we replace the dataset with the vector $(S_1, \ldots, S_{\nbInd})$.

Given $T_i=1$, resp. $T_i=0$, the technical replicates $X_{ij}$ are independent Bernoulli random variables with parameter $1-q$, resp. $p$. Thus, $S_i$ given the true value $T_i$ is
\begin{equation} \label{eq:SgivenT}
    \big[S_i\big|T_i\big]\sim \mathscr{B}\text{in}\big(n_i, T_i(1-q)+(1-T_i)p\big)=
    \ind{T_i=1} \mathscr{B}\text{in}(n_i, 1-q) + \ind{T_i=0}\mathscr{B}\text{in}(n_i, p)
\end{equation}
with independence between the $S_i$'s given $(T_1,\ldots, T_{\nbInd})$.

\subsection{Scoring technical replicates with the mean and the median} \label{sec:scoring_AM}

The average-based score $Y_{A,i}$ is defined as
\[
Y_{A,i} = \frac{1}{n_i} \sum_{j=1}^{n_i} X_{ij}=\frac{S_i}{n_i}.
\]
This score is heavily used in practice, as it is a simple scaling of the sufficient statistic $S_i$.\\
The median-based score $Y_{M,i}$ is defined as
\begin{equation} \label{eq:Y_M1}
    Y_{M,i} =\operatorname{median}\big(X_{i1},\ldots, X_{in_i}\big)=\ind{S_i > n_i/2} + \frac12 \ind{S_i = n_i/2}.
\end{equation}
In the specific case where $n_i$ is even and when there is a tie in the frequencies of $0$ and $1$ within the replicates of individual $i$, we set $Y_{M, i}=1/2$. \\
Note that with the above definitions, the mean and median scores are equal when $n_i\in\{1,2\}$. However, their effectiveness is not guaranteed, and we will compare them with other scoring methods that are non-linear functions of $S_i$.

\subsection{Maximum-A-Posteriori scoring}\label{sec:EM}
The most likely value of $T_i$ given the observed values of $S_i$ depends on the value of $\p(T_i=1|S_i=s_i)$, whether below or above $1/2$.
Using~\eqref{eq:SgivenT}, $T_i\sim\mathscr{B}\text{er}(\theta_T)$ and the Bayes formula, the distribution of $T_i$ given $S_i$ is
\begin{equation}
    \big[T_i\big|S_i=s_i\big] \sim \mathscr{B}\text{er}\left(\frac{\theta_T q^{n_i-s_i} (1-q)^{s_i}}{\theta_T q^{n_i-s_i} {(1-q)}^{s_i} + (1-\theta_T)p^{s_i} {(1-p)}^{n_i-s_i}}\right).
    \label{eq:TgivenS}
\end{equation}
Thus, we introduce the likelihood-based score $Y_{L, i}(\theta_T, p, q)$ as
\begin{equation} \label{eq:Y_L}
    Y_{L,i}(\theta_T, p, q) = \p(T_i=1|S_i=s_i) = \frac{\theta_T q^{n_i-s_i} {(1-q)}^{s_i}}{\theta_T q^{n_i-s_i} (1-q)^{s_i} + (1-\theta_T)p^{s_i} {(1-p)}^{n_i-s_i}}.
\end{equation}
This score is an increasing function of $s_i$ because of Lemma~\ref{lem:Lincrease} in~\ref{sec:Lincrease}. This score $Y_{L, i}(\theta_T, p, q)$ depends on the values of $p$, $q$, and $\theta_T$. Estimating the fixed parameters $\theta_T$, $p$, and $q$ by maximum likelihood is possible. Yet, we need to penalize the likelihood to obtain non-degenerate estimates.

The likelihood of the fixed parameters $\theta_T$, $p$, and $q$ given the data $(s_1,\ldots, s_{\nbInd})$ 
is
\begin{equation}\label{eq:likelihood}
    \Lik(\theta_T, p, q) = \prod_{i=1}^n \Big\{\theta_T (1-q)^{s_i}q^{n_i-s_i} + (1-\theta_T) p^{s_i}{(1-p)}^{n_i-s_i}
    \Big\},
\end{equation}
up to a multiplicate constant.
Besides, this maximum might suffer from unrealistic maxima, corresponding to $p=0$ or $q=0$, for example. For this reason, we penalize the likelihood with a $\mathscr{B}\text{eta}(2,2)$-prior on both $p$ and $q$ to avoid $p=0$ or $q=0$.. The maximum of the corresponding functional, the Maximum-A-Posteriori (MAP), is not explicit. As for maximum likelihood, we need to resort to numerical optimization. The standard algorithm to fulfill this task in mixture models is the Expectation-Maximization (EM) algorithm; see, e.g., Chapters 1, 2, and 9 of~\cite{fruhwirth2019handbook}. For further details, we refer to~\ref{app:EM}.

By plugging in the result $(\widehat{\theta}_{T,\MAP},\widehat{p}_{\MAP},\widehat{q}_{\MAP})$ of the EM algorithm into \ref{eq:Y_L}, we define the MAP score ${Y}_{\text{MAP}, i}$ as
\begin{align}\label{equ:mlScore}
    {Y}_{\MAP,i} = Y_{L,i}(\widehat{\theta}_{T, \MAP},\widehat{p}_{\MAP},\widehat{q}_{\MAP}).
\end{align}

\subsection{Bayesian scoring}\label{sec:Bayesian_scoring}

\begin{figure}[tbh]
    \centering
    \begin{minipage}{.45\textwidth}
        \includegraphics[width=\textwidth]{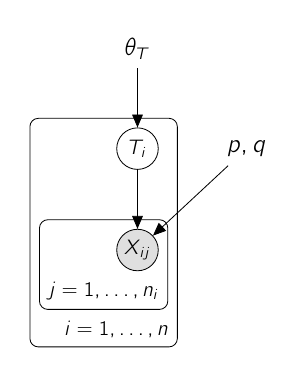}
    \end{minipage}%
    \begin{minipage}{.53\textwidth}
        \begin{align*}
            \theta_T &\sim \mathscr{B}\text{eta}(a_T,b_T),\\
            p &\sim \mathscr{B}\text{eta}(a_\text{FP},b_\text{FP})_{|p\in(0,1/2)},\\
            q &\sim \mathscr{B}\text{eta}(a_\text{FN},b_\text{FN})_{|q\in(0,1/2)},\\
            \big[T_i\big|\theta_T\big]
            &\sim \mathscr{B}\text{er}(\theta_T),\\
            \big[X_{ij}\big|T_i,p,q\big]
            &\sim \mathscr{B}\text{er}\big(T_i(1-q) + (1-T_i)p\big),\\
            \big[S_i\big|T_i, p, q\big]
            &\sim \mathscr{B}\text{in}\big(n_i,T_i (1-q)+(1-T_i)p\big),
            \\
            \text{where } S_i & = \sum_{j=1}^{n_i} X_{ij}.
        \end{align*}
    \end{minipage}
    \caption{\textbf{The Bayesian model}: the Directed Acyclic Graph (left) and the generative model (right). Variables $\theta_T$, $p$, and $q$ are fixed parameters (with a prior distribution), the $T_i$'s are latent variables, and the $X_{ij}$'s are the observed data. The hyperparameters $a_T,b_T, a_\text{FP}, b_\text{FP}, a_\text{FN}, b_\text{FN}$ set the prior distribution and should be chosen by the user. The prior distribution on $p$ and $q$ are Beta distributions truncated so that $p$ and $q$ are both in $(0,1/2)$. The loop over the replicates $j=1,\ldots, n_i$ can be replaced by a single node $S_i$, counting the number of ones among the $X_{ij}$'s since it is a sufficient statistic.}\label{fig:DAG}
\end{figure}

We encompass the stochastic model described above in a Bayesian framework. We consider a Bayesian model with a prior distribution on $\theta_T$, $p$, and $q$ whose density is denoted $\pi(\theta_T, p,q)$, see Figure~\ref{fig:DAG}. In the Bayesian framework, $\p_\pi$ and $\e_\pi$ denote the probability measure and the expected value, as detailed in \ref{app:bayesian}

Using the Bayesian model defined above and its posterior distribution, the Bayesian score $Y_{B, i}$ is defined as the following posterior expected value
\[
    Y_{B,i} = \e_\pi\big(T_i\big|S_1,\ldots,S_{\nbInd} \big)=\p_\pi\big(T_i=1\big|S_1,\ldots,S_{\nbInd} \big),
\]
that integrates $T_i$ over the posterior values of $\theta_T$, $p$, and $q$ and thus over the uncertainty on the fixed parameters. Alternatively, assuming $\pi(\theta_T, p,q|S_1,\ldots, S_{\nbInd})$ is the posterior distribution of the fixed parameters given the data, it can be viewed as the posterior expected value of the likelihood-based score $Y_{L, i}(\theta_T,p,q)=\p(T_i=1|S_i)$, namely
\begin{equation}\label{eq:Y_B}
    Y_{B,i} = \int Y_{L,i}(\theta_T, p, q)
    \pi(\theta_T, p, q|S_1,\ldots, S_{\nbInd})\mathrm{d}\theta_T\mathrm{d}p\mathrm{d}q.
\end{equation}
Unlike scores $Y_{A, i}$ and $Y_{M, i}$, values of $Y_{\MAP, i}$ and $Y_{B, i}$ depend not only on the replicates observed on the $i$-th individual (i.e., $S_i$) but also on the whole dataset through the posterior distribution on the fixed parameters.

To decide how to set the hyperparameters for the prior distribution of the fixed parameters $\theta_T$, $p$, and $q$, see Figure~\ref{fig:DAG}, remember that each of these parameters represents the probability of getting a $1$ in a binary ($0$/$1$) trial. For these probabilities, we use the Beta distribution $\mathscr{B}\text{eta}(a, b)$, where $a, b > 0$. This distribution represents information from a sample of size $a + b$ consisting of $a$ observations equal to $1$ and $b$ equal to $0$. 
A non-informative prior is the Beta distribution with $a = b = 1/2$, representing minimal prior knowledge. This can be thought of as a fictive sample of size $1$ with equal likelihoods for $0$ and $1$. We recommend using this non-informative prior for the prevalence $\theta_T$. The uniform distribution on $(0, 1)$ is a Beta distribution with $a = b = 1$, providing a weakly informative prior based on a balanced sample of size $2$.
For the false-positive rate $p$ and the false-negative rate $q$, setting $a_{\text{FP}} = a_{\text{FN}} = b_{\text{FP}} = b_{\text{FN}} = 1/2$ is not ideal: this distribution places too much weight near $0$, where the likelihood-based scoring performs poorly. Since replicates imply noisy measurements, meaning $p$ and $q$ should not be too close to $0$, we recommend using $a_{\text{FP}} = a_{\text{FN}} = b_{\text{FP}} = b_{\text{FN}} = 2$, which puts less weight near $0$.
In summary, we propose using $a_T = b_T = 1/2$ and $a_{\text{FP}} = a_{\text{FN}} = b_{\text{FP}} = b_{\text{FN}} = 2$ as default values for the hyperparameters, adjusting them as needed if more information is available. See Figure~\ref{fig:periodontal_posterior}, Section~\ref{sec:periodontal}, for an example.

Finally, the median-based score is a limit case of the Bayesian score. Specifically, when the prior is highly concentrated around $(\theta_T, p, q) = (1/2, 0, 0)$, the posterior distribution of the fixed parameters also stays concentrated around $(1/2, 0, 0)$. In this scenario, the Bayesian score $Y_{B,i}$ defined in~\eqref{eq:Y_B} becomes approximately equal to $Y_{L,i}(1/2, 0, 0)$, which corresponds to the median-based score $Y_{M,i}$. Thus, the median-based score can be seen as an approximation of the Bayesian score when using a strongly informative prior.
This strongly informative prior assumes that the false-positive and false-negative rates are negligible and that the prevalence is roughly $1/2$. Such a situation occurs with the prior defined in Figure~\ref{fig:DAG}, when the hyperparameters $b_{\text{FN}} = b_{\text{FP}} \to \infty$ and $a_T = b_T \to \infty$ simultaneously.

\subsection{From scorings to prevalence and error rates estimators}\label{sec:scoring_to_estimates}

To infer the prevalence $\theta_T$, we consider the four following estimates:
\begin{equation}
\label{Eq:PrevEst}
    \forall K\in\{A, M, \MAP\},\ \widehat{\theta}_{T,K} =  \dfrac{1}{{\nbInd}}\sum_{i=1}^{\nbInd} Y_{K,i}\quad\text{and}\quad
    \widehat{\theta}_{T,B} = \e_\pi\big(\theta_T\big|S_1, \ldots, S_{\nbInd}\big).
\end{equation}
For average-, median- and maximum-a-posteriori-based methods, scorings can be used to estimate the false-positivity and false-negativity rates $p$ and $q$:
\begin{equation}
\label{Eq:pqEst}
    \forall K\in\{A, M, \MAP \},\ \widehat{p}_{K} =  \frac{\sum_{i=1}^{\nbInd} S_i (1-{Y}_{K,i}) }{\sum_{i=1}^{\nbInd} n_i (1-{Y}_{K,i}) } \ \textrm{ and } \ \widehat{q}_{K} =  \frac{\sum_{i=1}^{\nbInd}  (n_i-S_i){Y}_{K,i}}{\sum_{i=1}^{\nbInd} n_i {Y}_{K,i}}  .
\end{equation}

Bayesian statistics has its way of estimating parameters using their posterior expected values. We can approximate them easily by computing the average of a sample of $\theta_T, p$ or $q$ drawn from the posterior distribution with the Hamiltonian Monte Carlo algorithm of \texttt{rstan}.

\subsection{From scorings to classifiers}\label{sec:scoring_to_classifier}

We compute our predictions of the latent $T_i$ values by thresholding the scores. 
The above scoring statistics take values in $[0,1]$ and summarize information from a given dataset, the $s_i$s, into a single value.
We should interpret them as a tool to infer the latent $T_i$ value: when $T_i=1$, we expect high scores; when $T_i=0$, we expect low scores.
We introduce an indecision response ($1/2$) when we would not trust a decision based on the observed replicates. Since binary data carry little information, this could happen. It is an invitation to add more replicates related to this individual before deciding.

To classify the individuals, we introduce two thresholds $0 < v_L \leq 1/2 \leq v_U < 1$ and set the classifiers as
\begin{equation}\label{eq:classifier}
    \widehat{T}_{K,i} = \Phi(Y_{K,i}), \quad \text{where}\ \Phi(y)= \begin{cases}
    0 & \text{if } y < v_L,\\
    1/2 & \text{if } v_L \leq y \leq v_U,\\
    1 & \text{if } y > v_U,
    \end{cases}
\end{equation}
for all methods $K\in\{A, M, \MAP, B\}$ and all individuals $i=1,\ldots, n$.
Since the median-based score $Y_{M,i}$ is always in $\{0, 1/2, 1\}$, we always have $\widehat{T}_{M,i} = Y_{M,i}$. This double thresholding method is related to risk theory. To compute the risk of a classifier, we introduce a loss function $\ell(t, \widehat t)$ that quantifies the cost of predicting $\widehat t\in\{0, 1/2, 1\}$ when the truth is $t\in\{0, 1\}$. See Section~\ref{sec:loss} for more details.

\subsection{Predictions}

Let us assume that a new individual $n+1$ is given to the agent through $n_{n+1}$ and $s_{n+1}$. It is possible to give the posterior prediction of its score $\widehat{Y}_{n+1}$ based on the dataset composed by the $n$ previous individuals. For $K\in\{A, M, \MAP\}$, we estimate the parameters $\widehat{\theta}_{T,K}, \widehat{p}_{K}$ and $ \widehat{q}_{K}$, from the individuals $\{1, \cdots, n \}$, and we compute
\begin{equation}
\label{Eq:PredAMMAP}
    \widehat{Y}_{K,n+1} = Y_{L,n+1}(\widehat{\theta}_{T,K}, \widehat{p}_{K}, \widehat{q}_{K}),
\end{equation}
where $\dis Y_{L,n+1}(\theta_T, p, q)=\p(T_{n+1}=1|S_{n+1}=s_{n+1})$ is computed as in Equation~\eqref{eq:Y_L}. It is also possible to build the prediction score for the Bayesian approach, for which the form is
\begin{equation}
\label{Eq:PredB}
     \widehat{Y}_{B,n+1} = \int Y_{L,n+1}(\theta_T, p, q) \, \pi(\theta_T, p, q|S_1,\ldots,S_n) \, \text{d}\theta_T \, \text{d}p \, \text{d}q.
\end{equation}
A Monte-Carlo estimator is chosen to approximate the previous integral, such as
\begin{align*}
    \widehat{Y}_{B,n+1} = \frac{1}{H} \sum_{h=1}^H Y_{L,n+1}(\theta_{T,h}, p_h, q_h),
\end{align*}
where each parameter $(\theta_{T,h}, p_h, q_h)$ is sampled from the posterior distribution $\pi(\theta_{T}, p, q|S_1,\ldots,S_n)$ such as described in Section~\ref{sec:Bayesian_scoring}.

Finally, note that, for $K\in\{A, M, \MAP\}$, K-prediction scores do not include any variability over the parameters, while Bayesian-prediction scores do. Thus, these K-prediction scores might lead to over-confident decisions.

\section{Theoretical results}\label{sec:theoretical}

In this Section, we provide efficiency results that compare the various methods introduced above. We start with the efficiency of the classifiers in terms of sensitivity and specificity and of specific loss functions, as introduced in Section~\ref{sec:loss}, that deal with indecision responses on the replicates.
Section~\ref{sec:accuracy_classifiers} gives the results on the classifiers. Section~\ref{sec:accuracy_prevalence} gives the results on the prevalence estimators.
To state the results, we may need the following hypotheses.
\begin{enumerate}[(H1)]
    \item The false-positivity and false-negativity rates $p$ and $q$ are in $(0,1/2)$. \label{hhalf}
    \item There exists at least one individual for which the number of replicates $n_i\ge 3$. \label{h3replicates}
    \item The $v_L$ and $v_U$ that define the classifiers in Section~\ref{sec:scoring_to_classifier} are such that $0<v_L\leq 1/2 \leq v_U <1$. \label{hv}
    \item The loss function $\ell(t,\widehat t)$ satisfies the conditions~\eqref{eq:lemma1} of Lemma~\ref{lem:loss} in Section~\ref{sec:loss}. \label{hloss}
\end{enumerate}

\subsection{
Loss functions for classifiers with indecision response and minimal risk classifiers
} \label{sec:loss}

In this Subsection, we only look at the following problem.
Assume we want to predict a binary random $T\in\{0, 1\}$, with three possible decisions: $0$, $1/2$ (inconclusive) and $1$ based on the known value of $\vartheta=\p(T=1)$. (In the following section, $\vartheta$ can be $\theta_T$ or $\p(T_i=1|S_i)$ if we reason given the data.)

Consider the loss function $\ell(t,\widehat t)$, defined on $\{0,1\}\times\{0,1/2,1\}$ such as
\begin{center}
    \begin{tabular}{c|ccc}
    \hline
        $\ell(t,\widehat t)$ & $0$ & $1/2$ & $1$\\
        \hline
            $0$ & $0$ & $a$ & $b$\\
            $1$ & $c$ & $d$ & $0$\\
            \hline
    \end{tabular}
\end{center}
where $a$, $b$, $c$ and $d$ are positive constants.
When $a=b=1$ and $c=d=0$, the loss function is related to the specificity. When $c=d=1$ and $a=b=0$, the loss function is related to the sensitivity. And, when $b=c=1$ and $a=d=0$, the loss function is the misclassification error. The general loss function $\ell(t,\widehat t)$ can be interpreted as follows.
If $a$ and $d$ are small enough compared to $b$ and $c$, the indecision response $1/2$ may be the best choice in case of a substantial uncertainty between $0$ and $1$. In medical applications, asking for further tests may cost less than making a doubtful decision. An indecision response is always an error, whether the truth is $0$ or $1$. Therefore, 
decision methods 
propose $\widehat{t}=1/2$ only if indecision costs are sufficiently low. The constraints are given in Equation~\eqref{eq:lemma1} of Lemma~\ref{lem:loss}. This gives $a<1/2$. However, this indecision cost $a$ must remain high (i.e. not too close to 0) for decisions to be made in most cases.

Sometimes we would consider $\ell$ in symmetrical form, taking $b=c=1$ and $a=d$, denoted ${\ell}_a$,
where $0<a<1$ is the indecision cost. Thus the cost of a false positive, $\ell_a(0,1)$, is equal to the cost of a false negative, $\ell_a(1,0)$. This is unrealistic in the context of medical diagnosis. Indeed, we often wish to avoid false negatives not to leave a diseased patient untreated. In this case, an absence of decision with $\widehat{t}=1/2$ is better than the false negative. The resulting cost of indecision, $a$, is thus lower than that of a false negative, set to $1$.

For a given loss function $\ell$, we are interested in the best classifiers $\widehat T^\star$ that minimize the risk $r(\widehat t)=\e\big(\ell(T,\widehat t)\big)$, defined as
\[
\widehat T^\star = \argmin_{\widehat t\in\{0, 1/2, 1\}} r(\widehat t),
\]
and we select the indecision response, i.e., $1/2$, if the risks tie in. Lemma~\ref{lem:loss} (a proof is given in~\ref{sec:loss_proof}) gives the conditions on $a, b, c$ and $d$ under which the indecision response can appear. It also provides the best decision in this case.
\begin{lem}\label{lem:loss}
    Assume $\vartheta=\p(T=1)$ is known.
    The best classifier $\widehat T^\star$ is
     \[
        \widehat T^\star = \Phi(\vartheta),
      \]
       where $\Phi$ is defined in Equation~\eqref{eq:classifier} with $v_L = \frac{a}{c-(d-a)}$ and $ v_U = \frac{b-a}{(d-a)+b}$
    \begin{align}
        \text{if and only if}\quad\frac{bc}{b+c} > a + (d-a)\frac{b}{b+c}
        \quad \text{and} \quad
        -b < (d-a) < c.\label{eq:lemma1}
    \end{align}
\end{lem}

If we apply Lemma~\ref{lem:loss} to the symmetric loss function $\ell_a$, we must choose $v_L=a$ and $v_U=1-a$ to obtain the best classifier.

\subsection{Accuracy of the classifiers as a diagnostic tool}\label{sec:accuracy_classifiers}
Whatever the scoring method, we rely on the same two thresholds $v_L \leq v_U$ to transform the scores into diagnostics. Let us denote
\[
n_0 = \max_i n_i, \quad \text{and} \quad \delta_0= \begin{cases}\frac{1}{2 n_0}  & \text{if $n_0$ is odd},
\\ \frac{1}{2 (n_0-1)} &\text{otherwise.}
\end{cases}
\]
In the limit case where $v_L \in \left[ \frac{1}{2} - \delta_0; \frac{1}{2} \right]$ and $v_U \in \left[ \frac{1}{2} ; \frac{1}{2} + \delta_0 \right]$, it is easy to prove that both average- and median-based classifications are identical. Otherwise, we can compare the sensitivity and specificity of the average-based and median-based classifiers.

\begin{thme}\label{thme:accuracy_AM_classifiers}
    Assume {\em(H\ref{hhalf})} and {\em(H\ref{hv})}. We have
    \[
        \operatorname{sensitivity}(\widehat T_{A,i}) \le \operatorname{sensitivity}(\widehat T_{M,i}).
    \]
    The above inequality is strict if and only if $v_U > 1/2$, becoming $v_U > 1/2 + \delta_0$. Moreover, we have
    \[
        \operatorname{specificity}(\widehat T_{A,i}) \le \operatorname{specificity}(\widehat T_{M,i})
    \]
    The above inequality is strict if and only if $v_L < 1/2$, becoming $v_L < 1/2 - \delta_0$.

    As a consequence, the median-based classifier is also better in terms of the misclassification rate and informedness.
\end{thme}
The proof is given in \ref{sec:proof_accuracy_AM_classifiers}. The theorem states that the median-based classifier is better than the average-based classifier regarding sensitivity and specificity when we introduce the inclusive response, i.e., as soon as $v_L<0.5<v_U$.
Both classifiers reflect the properties of their respective scoring. Hence, Theorem~\ref{thme:accuracy_AM_classifiers} yields a first conclusion on the efficiency of the scoring methods in favor of the median-based scoring.

We also obtained efficiency results on the 
Bayesian classifier. To state it, we must refer to the loss function $\ell(t,\widehat t)$ defined in Section~\ref{sec:loss}. Bayesian statistics, which is well grounded in decision theory, is known to be efficient in the sense that it provides the best possible estimators and classifiers given the data if the statistics are computed wisely with the posterior distribution, see, e.g.,~\cite{robert2007bayesian}. Here, we can prove the following results.
\begin{thme}\label{thme:accuracy_bayesian_classifiers}
    Assume {\em(H\ref{hhalf})} and {\em(H\ref{hloss})}. Consider any classifier $\widehat T_i$ of the $i$-th individual that returns a decision in $\{0, 1/2, 1\}$, based on the data $S_1,\ldots, S_n$.
    \begin{enumerate}[(i)]
        \item \label{item1} {\em (Optimality of the likelihood-based classifier)}
        Whatever the values of $(\theta_T, p, q)$, we have
        \[
            \e(\ell(T_i,\widehat T_{L,i}(\theta_T,p,q))) \le \e(\ell(T_i,\widehat T_{i})).
        \]
        \item \label{item2}
        {\em (Bayesian optimality of the Bayesian classifier)}
        We have
        \[
            \e_\pi(\ell(T_i,\widehat T_{B,i})) \le \e_\pi(\ell(T_i,\widehat T_{i})).
        \]
        \item \label{item3}{\em (Admissibility of $\widehat T_{B,i}$)}
        If on a set of values of $(\theta_T,p,q)$ with positive prior probability we have
        \[
            \e\Big[\ell\big(T_i,\widehat T_{i}\big)\Big] <
            \e\Big[\ell\big(T_i,\widehat T_{B,i}\big)\Big],
        \]
        then there exists another set of values of $(\theta_T, p, q)$ with positive prior probability for which the inequality is reversed (strictly).
    \end{enumerate}
\end{thme}
The proof is given in~\ref{sec:proof_accuracy_Bayesian_classifiers}. Note that $\e$ is the expected value given the fixed parameters $\theta_T$, $p$, and $q$, whereas $\e_\pi$ integrates them according to the prior distribution of density $\pi$. Item~\eqref{item1} states that the likelihood-based classifier is optimal regarding risk, i.e., expected loss at fixed values of $\theta_T$, $p$, and $q$.
Item~\eqref{item2} considers the Bayesian risk, which is the expected loss integrated over the prior distribution of the fixed parameters. The Bayesian classifier is optimal in this sense.
As stated in item~\eqref{item3}, the Bayesian classifier is admissible: no other classifier can outperform it uniformly over the entire set of fixed parameter values. Because of these results, we recommend using the Bayesian classifier in practice.

\subsection{Accuracy of the prevalence estimators}\label{sec:accuracy_prevalence}
As defined in Equation~\eqref{Eq:PrevEst}, we have four prevalence estimators: the average-based, the median-based, the Maximum-A-Posteriori, and the Bayesian prevalence estimators. The latter $\widehat\theta_{T, B}$ is the expected value of the posterior distribution of the prevalence $\theta_T$ given the data. In contrast, the former three $\widehat\theta_{T, K}$, with $K = \{A, M, \MAP\}$, are the empirical means of the associated scores $Y_{K, i}$ for $i=1, \cdots {\nbInd}$.

We first consider the two empirical means $\widehat\theta_{T, A}$ and $\widehat\theta_{T, M}$. They are heavily biased, with a bias that does not tend to $0$ as the number of individuals increases. On the other hand, their variances are proportional to $1/n$, see~\ref{sec:prevalence_variance}. Thus, asymptotically, their squared biases dominate their mean squared errors. Moreover, we can compare their bias as follows.
\begin{thme}
\label{thme:bias_AM}{\em (Bias of $\widehat\theta_{T, A}$ and $\widehat\theta_{T, M}$)}
    Assume {\em(H\ref{hhalf})} and set $n_0 = \min\{n_i, i=1,\ldots,{\nbInd}\}$.

    For any  values of $p$ and $q$ in $(0,1/2)$, there exists an interval $J$ that contains $p/(p+q)$ such that
    \[
        |\e\big(\widehat\theta_{T, M}\big) - \theta_T| \le |\e\big(\widehat\theta_{T, A}\big) - \theta_T| \, ,
    \]
    except if $\theta_T\in J$.

    The bias of $\widehat\theta_{T, A}$ is not influenced by the numbers of replicates $n_i$'s.
    Whereas, as $n_0\to\infty$, then the bias of $\widehat\theta_{T, M}$ tends to $0$ and
    the length of the interval $J$ tends to $0$.
\end{thme}
The proof is given in~\ref{sec:proof_prevalence}. The theorem states that the median-based prevalence estimator $\widehat\theta_{T, M}$ is better than the average-based prevalence estimator $\widehat\theta_{T, A}$ in terms of bias, except when $\theta_T$ is in an interval $J$. And the length of $J$ is small when the number of replicates is always large. In the latter case, we can rely on the median-based $\widehat\theta_{T, M}$ to estimate the prevalence. Otherwise, both estimators are heavily biased and should be used with caution.

As always, Bayesian statistics come with its efficiency. In terms of mean squared error, we can prove the following results on the Bayesian prevalence estimator $\widehat{\theta}_{T, B}$.
\begin{thme}\label{thme:accuracy_Bayesian_prevalence}
    Assume {\em(H\ref{hhalf})}, and consider an estimator $\widehat\theta_T$ of $\theta_T$, that is to say any function of the data $S_1,\ldots, S_{\nbInd}$.
    \begin{enumerate}[(i)]
        \item {\em (Bayesian optimality of $\widehat \theta_{T,B}$)}
        We have
        \[
            \e_\pi\Big[(\widehat \theta_{T,B} -\theta_T)^2\Big] \le \e_\pi\Big[(\widehat \theta_T -\theta_T)^2\Big].
        \]
        \item {\em (Admissibility of $\widehat \theta_{T,B}$)}
        If on a set of values of $(\theta_T,p,q)$ with positive prior probability we have
        \[
            \e\Big[\big(\widehat \theta_T -\theta_T\big)^2\Big] <
            \e\Big[\big(\widehat \theta_{T,B} -\theta_T\big)^2\Big],
        \]
         then there exists another set of values of $(\theta_T, p,q)$ with positive prior probability for which the inequality is reversed (strictly).
    \end{enumerate}
\end{thme}

The proof is given in~\ref{sec:proof_accuracy_Bayesian_prevalence}.
The above Theorem states that $\widehat\theta_{T, B}$ is optimal regarding Bayesian $L^2$-risk, which is the mean squared error integrated over the prior distribution. Moreover, the Bayesian estimator is admissible, which means that there is no other statistic whose mean squared error is always smaller than the one of $\widehat\theta_{T, B}$ whatever the values of the fixed parameters. Additionally, the Bayesian methodology evaluates the uncertainty of the estimated value with credible intervals. These arguments favor the Bayesian prevalence estimator and recommend its use in practice.

\section{Numerical results}
\label{Sec:NumRes}

To evaluate the performance of a method on a specific dataset, we use the empirical ${\overline{\ell}_a}$-risk, defined as
\begin{align*}
    {\overline{\ell}_a}\textrm{-risk} &= \frac{1}{{\nbInd}}  \sum_{i=1}^{\nbInd} {\ell}_a (t_i,\widehat{t}_i),
\end{align*}
\textit{i.e.}, the average of the losses we commit with all decisions taken for each observation $i$, where the loss function ${\ell}_a$ was defined in Section \ref{sec:loss}.

\subsection{Some simulations}
\label{sec:Toy}
\begin{figure}[tb]
    \centering
    \begin{tabular}{cc}
    \rotatebox[origin=c]{90}{$\widehat{\theta}_{T,K}-\theta_T$\hspace{-9cm}}&
    \includegraphics[width=.9\textwidth]{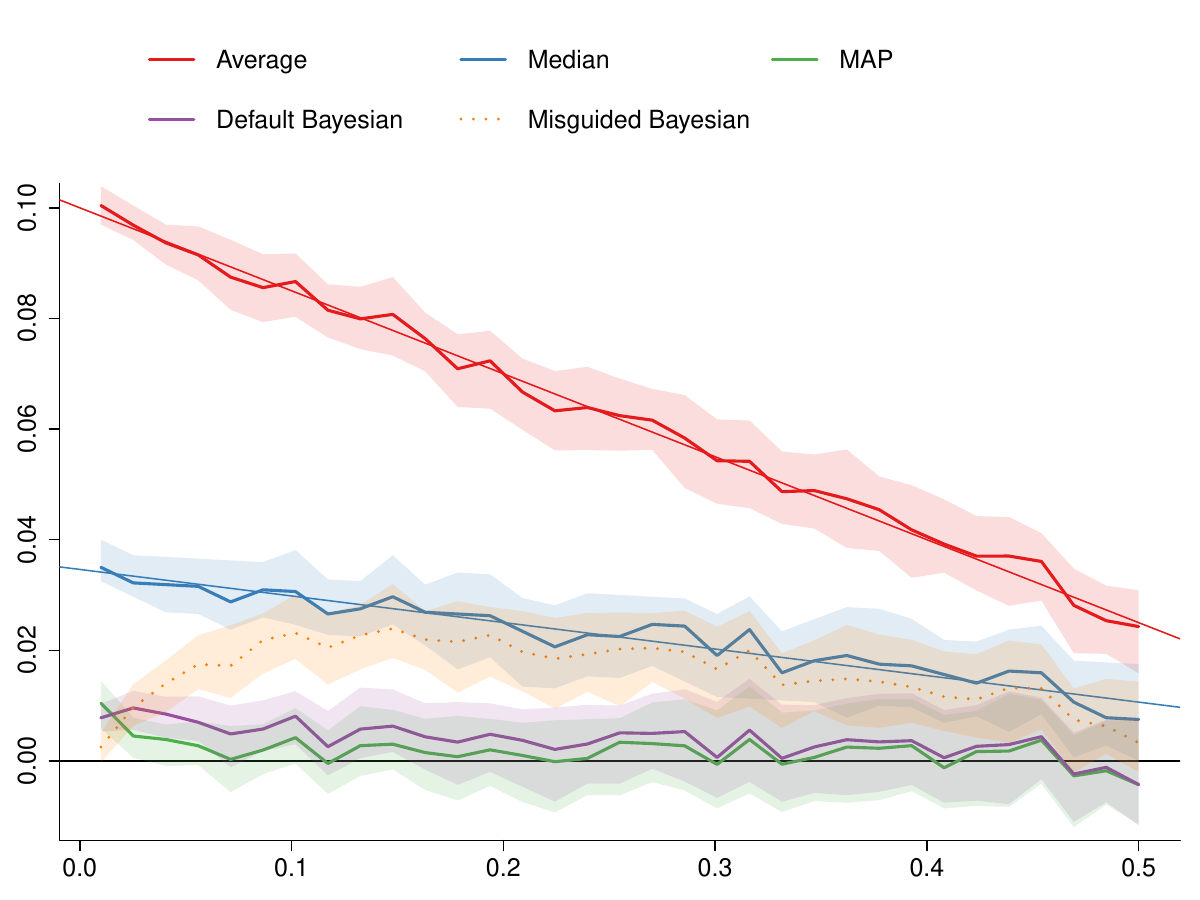}\\
    & $\theta_T$
    \end{tabular}
    \caption{Bias in prevalence estimation for simulated datasets, versus $\theta_T\in[0.01,0.5]$. Medians are plotted in thick lines, and $[0.4,0.6]$-quantile areas are filled with shaded colors, \numberSampleVarThe~datasets have been sampled for each value given to $\theta_T$. For average and median-based approaches, the theoretical bias, computed in \ref{sec:proof_prevalence}, are plotted as thin straight lines. The horizontal black line corresponds with the objective: null error on estimating $\theta_T$.}
    \label{fig:Simulations_varying_thetaT}
\end{figure}
\begin{figure}[tbh]
    \centering
     \begin{tabular}{cc}
     \rotatebox[origin=c]{90}{${\overline{\ell}_a}$-risk\hspace{-9cm}} &
     \includegraphics[width=.9\textwidth]{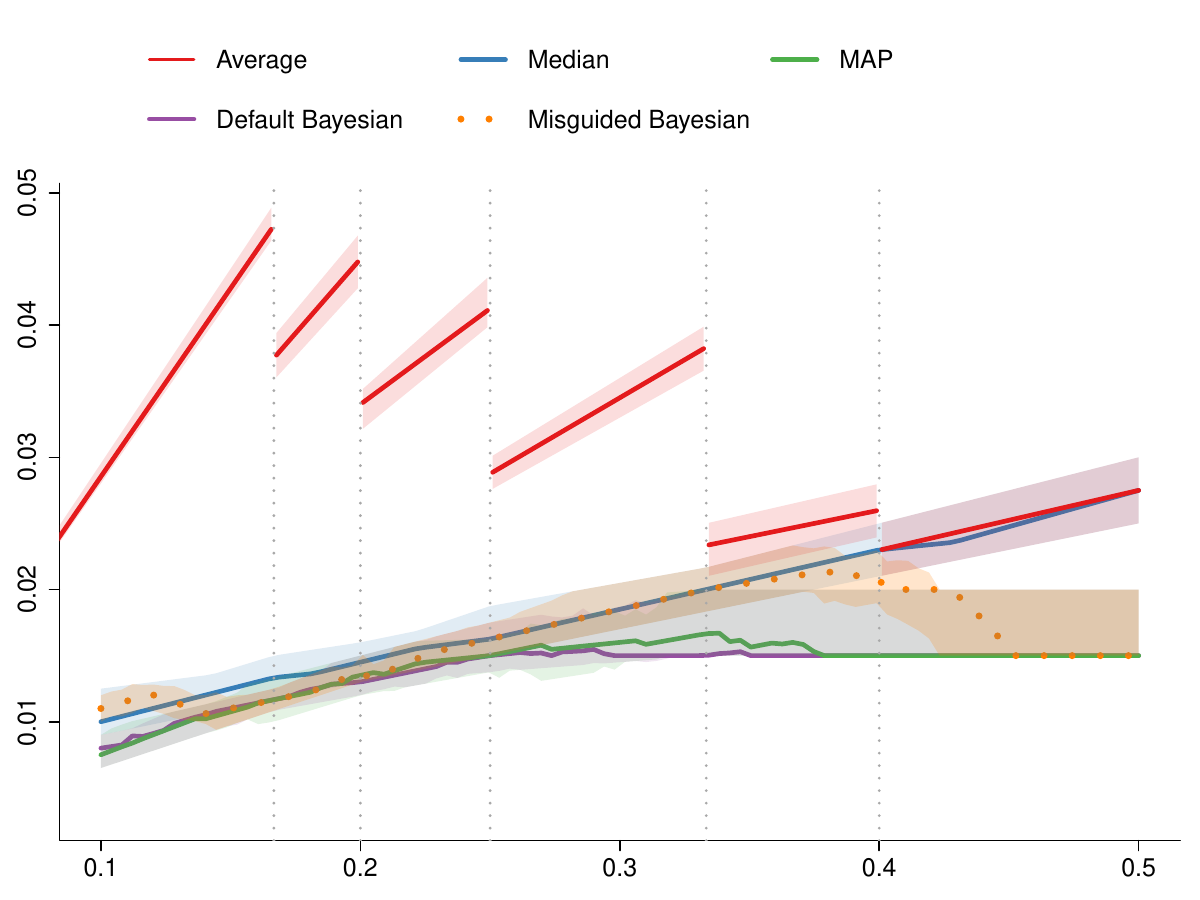}\\
      & $a$
     \end{tabular}
    \caption{Comparison of all methods when the decision cost $a$ varies and $\theta_T=0.4$. Their performance is measured in empirical ${\overline{\ell}_a}$-risk. We run the methods on $300$ simulated datasets for $a\in[0.1,0.5]$. The plain line is the median of the $300$ empirical risks, and the band represents the $[0.4,0.6]$-quantile interval. The average-risk jumps when $a$ is equal to an observed value $\frac{s_i}{n_i}$, where $ n_i \in [\![2,6]\!]$. In other words, jumps occur when $a$ crosses $\{\frac{1}{6}; \frac{1}{5}; \frac{1}{4}; \frac{1}{3}; \frac{2}{5} \}$, represented as vertical dotted lines.}
    \label{fig:ell_risk_simul}
\end{figure}

Simulations are based on the Bayesian model described in Figure~\ref{fig:DAG} where $p=0.1, q=0.05$ and $n_i$'s are sampled through a following uniform distribution $n_i\sim \mathscr{U}_{[\![2,6]\!]}$
and the Table~\ref{tab:priorSimusVarTheta} details how priors are built.

\begin{table}[b]
    \centering
    \begin{tabular}{rcccccc}
	\hline
      	&$a_{FP}$ & $b_{FP}$ & $a_\text{FN}$ & $b_\text{FN}$ & $a_\text{T}$ & $b_\text{T}$  \\
	\hline
        Default prior & 2&2&2&2& 0.5&0.5\\
        Misguided prior & 50&50& 50&50& 0.5&0.5\\
        \hline
    \end{tabular}
    \caption{Hyperparameters values of two different priors used to analyze the simulations. The first one is the default prior given in \ref{sec:Bayesian_scoring}. The second one is an informative prior that is badly chosen on purpose. See Figure~\ref{fig:DAG} for a description of the prior. }
    \label{tab:priorSimusVarTheta}
\end{table}

First, we compare the quality of estimating $\theta_T$ through the different approaches. Parameter $\theta_T$ is evenly sampled between $0.01$ and $0.5$. For each $\theta_T$, datasets of size ${\nbInd}=200$, have been sampled. Figure~\ref{fig:Simulations_varying_thetaT} gives the results of estimating the parameter $\theta_T$, where $\hat{\theta}_{T,K}$ is one of the estimators produced by any of the considered approaches. Medians (thick lines) and quartiles (shaded areas) are represented. The closer each curve is to the straight black line, the better the approach is. As expected from Theorem~\ref{thme:bias_AM}, the sample mean of median scores produces better prevalence estimates than the sample mean of average scores. Moreover, the bias obtained on simulations with the average-based and the median-based approaches follows the linear bias expected from \ref{sec:proof_prevalence}, plotted as thin straight lines. Furthermore, Maximum-A-Posteriori outperforms other approaches. The default Bayesian is close to these optimal methods, followed by the median-based approach. On the contrary, the prevalence value inferred by the average-based method is the worst-performing.

Next Figure~\ref{fig:ell_risk_simul} gives the $\overline{\ell}_a$-risks, for $\theta_T=0.4$, through their median and their $[0.4,0.6]$-quantile area in shaded.
As expected from \ref{sec:App_ellrisk}, the sample mean of median-risks is linear and increasing, whereas the sample mean of average-risks is linear and increasing, piecewise, where jumps occur at all the observed values of $\frac{s_i}{n_i} < 0.5$, when $n_i$ takes values in $[\![2,6]\!]$. In our simulations, the average-based approach gives the poorest results. The Bayesian and MAP approaches are always better than the average and the median solutions. The default Bayesian performs better than the misguided.

In Section~\ref{sec:accuracy_prevalence}, we recommended using the Bayesian prevalence estimator. The numerical results are in agreement with this recommendation.


\subsection{A periodontal dataset}
\label{sec:periodontal}

\begin{figure}
    \centering
    \includegraphics[width=\textwidth]{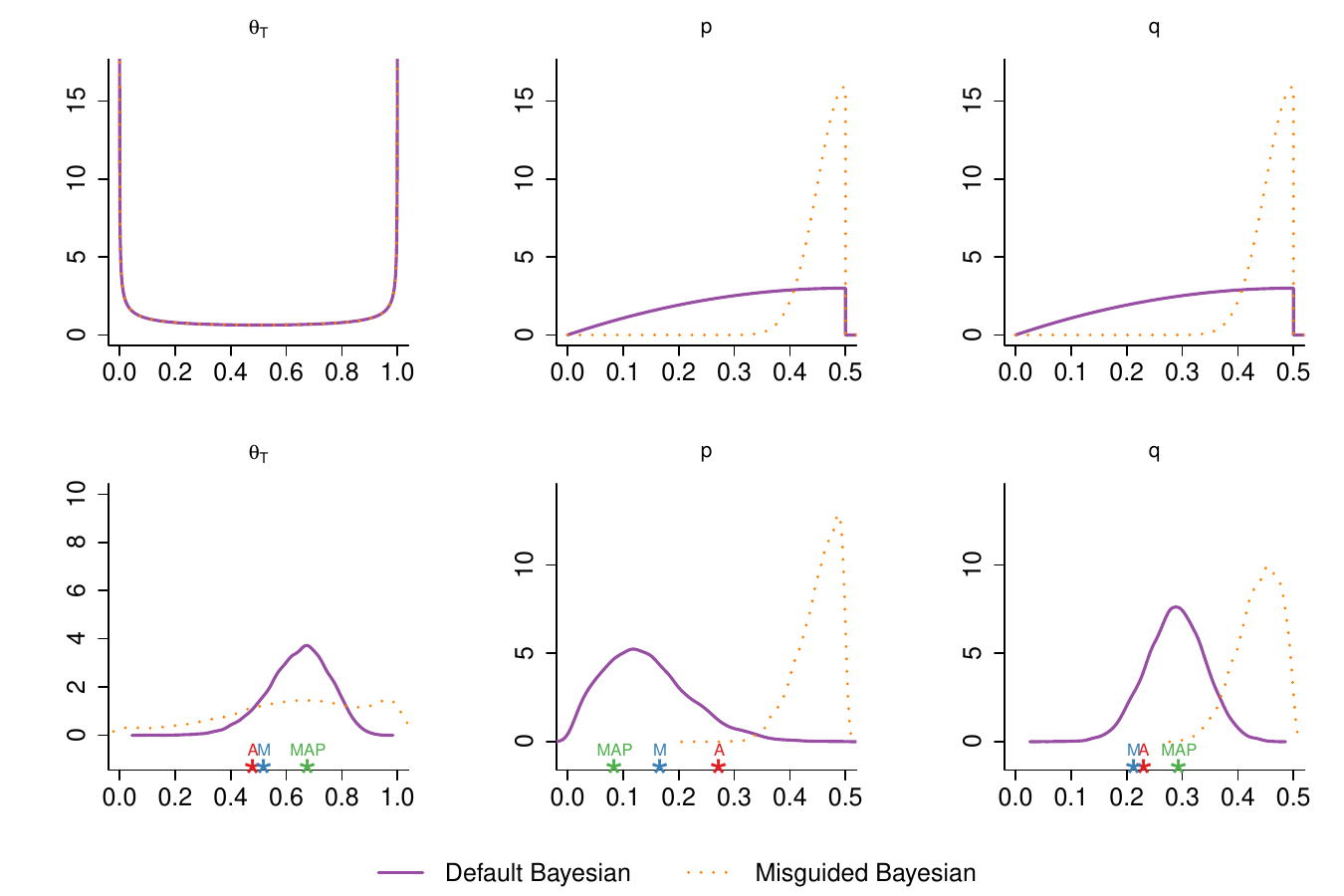}
    \caption{Comparison of the prior distributions (first row) and the associated posterior distributions (second row). From the latent status of each patient, we infer the parameters $\widehat{\theta}_{T, \theo} = 0.58, \ \widehat{p}_{\theo} = 0.187$ and $\widehat{q}_{\theo} = 0.338$. The parameters $\widehat{\theta}_{T, K}, \ \widehat{p}_{K}$ and $\widehat{q}_{K}$, estimated by non-Bayesian methods (K=A, M, \MAP) are displayed below the respective graphs with stars.}\label{fig:periodontal_posterior}
\end{figure}

We tried the proposed methods on the periodontal dataset of~\cite{HML90}, which includes $\nbInd=50$ individuals. The number of replicates per individual, $n_i$, varies from $1$ to $6$. The status variable of the dataset provides a trustworthy diagnostic (\textit{healthy} or \textit{infected}). We used it as the actual value of $T_i$. We ran (1) the average-based method, (2) the median-based method, (3) the Maximum-A-Posteriori estimator approximation as in Algorithm~\ref{algo:em}, (4) the default Bayesian method with the default prior given in Section~\ref{sec:Bayesian_scoring}, which is weakly informative and (5) a misguided Bayesian method, with a poorly chosen informative prior.
The prior and posterior distributions of Bayesian methods on the fixed parameters are displayed in Figure~\ref{fig:periodontal_posterior}.

Since we are in a unique situation where the latent status of each patient, denoted as $T_i$, is known, it becomes possible to infer the theoretical parameters $\theta_T$, $p$, and $q$ as
\begin{equation}
\label{Eq:ParEstTi}
    \widehat{\theta}_{T, \theo} =  \dfrac{1}{{\nbInd}}\sum_{i=1}^{\nbInd} T_{i}, \quad \widehat{p}_{\theo} =  \frac{\sum_{i=1}^{\nbInd} S_i (1-T_i)}{\sum_{i=1}^{\nbInd} n_i (1-T_i)} \ \textrm{ and } \ \widehat{q}_{\theo} =  \frac{\sum_{i=1}^{\nbInd} (n_i - S_i) T_i}{\sum_{i=1}^{\nbInd} n_i T_i}  .
\end{equation}
We can also infer the prevalence in the data set as $\widehat{\theta}_{T, \theo}  = 29/50=0.58$, the false positivity rate as $\widehat{p}_{\theo} = 9/48=0.187$ and the false negativity rate as $\widehat{q}_{\theo} = 48/142 =0.338$. Globally, we see that the default and the misguided Bayesian methods provide very different distributions, illustrating the influence of priors in Bayesian methods. The default Bayesian is almost centered on the values estimated from the latent status $T_i$. Indeed, the modes of the default Bayesian posterior density are $0.67, \, 0.12$ and $0.29$ for $\theta_T, p$ and $q$ respectively. On the contrary, the misguided Bayesian method provides a peaked posterior distribution for $p$ and $q$ that is wholly shifted towards $0.5$ (respective modes are $0.50$ and $0.46$). In contrast, it provides a uniform posterior distribution for $\theta_T$ estimation.

Furthermore, the estimation of parameters $p, q$ and $\theta_T$ provided by the average-based, the median-based and the Maximum-A-Posteriori methods, computed from Equations~\eqref{Eq:PrevEst} and \eqref{Eq:pqEst}, are represented in Figure~\ref{fig:periodontal_posterior} as points below the graphics. Contrary to Bayesian methods that provide a posterior distribution, these methods offer only a point estimate.

Next, we set the thresholds as $v_L=0.45$ and $v_U=1-v_L$ to compute the classifiers. This means an indecision response is given when the scores are too close to 1/2. We compared classifications to the diagnostics given by the status variable in the confusion matrix of Table~\ref{tab:periodontal}. Since the thresholds $v_L$ and $v_U$ are close to $0.5$, the average-based and the median-based classifiers give the same classification. Indeed, in this example we have $n_0=6$, then $a = 0.45 > 1/2 - \delta_0 = 0.4$, as defined in Section~\ref{sec:accuracy_classifiers}.

According to Lemma~\ref{lem:loss}, classification-accuracy is fitted with the loss $\ell_{0.45}(t,\widehat t)$. From the empirical risks given in Table~\ref{tab:periodontal}, we conclude that the default Bayesian method is the best, and the misguided Bayesian method is the worst. The misguided method has a bad $\ell_{0.45}$-risk because it diagnoses all patients as infected, whether they are infected or not. Default Bayesian and Maximum-A-Posteriori methods provide the same decisions, except for $3$ healthy individuals. For those three individuals, the default method is in favor of $1/2$ (inconclusive), whereas the Maximum-A-Posteriori method is in favor of $1$. Since predicting $1$ for a healthy patient is a larger error than predicting $1/2$, the empirical risk of the default Bayesian method ($10.2$) is smaller than that of the Maximum-A-Posteriori method ($11.9$).

\begin{table}[bt]
    \caption{Classification of the periodontal dataset. The table on the left counts how many times a decision has occurred given the method and the status of the individual. The thresholds were set to $v_L=a=0.45$ and $v_U=0.55$. The table on the right gives the empirical ${\ell}_{0.45}$-risk of each method, corresponding to $a=0.45$.}\label{tab:periodontal}
    \centering
    \vspace*{.2cm}
    \begin{tabular}{llrrr}
    \hline
    Status & Method & $\widehat T = 0$ & $\widehat T = 1/2$ & $\widehat T = 1$ \\
      \hline
Healthy & Average & 16 & 1 & 4 \\
  Healthy & Median & 16 & 1 & 4 \\
  Healthy & MAP & 13 & 0 & 8 \\
  Healthy & Default & 13 & 3 & 5 \\
  Healthy & Misguided & 1 & 7 & 13 \\
  \hline
  Infected & Average & 5 & 5 & 19 \\
  Infected & Median & 5 & 5 & 19 \\
  Infected & MAP & 3 & 2 & 24 \\
  Infected & Default & 3 & 2 & 24 \\
  Infected & Misguided & 1 & 4 & 24 \\
       \hline
    \end{tabular}%
    \hspace{2cm}
    \begin{tabular}{lc}
        \hline
        \multirow{2}{*}{Method}
        & Emp.\\ & $\ell_{0.45}$-risk \\
        \hline
Average & 11.7 \\
  Median & 11.7 \\
  MAP & 11.9 \\
  Default & 10.2 \\
  Misguided & 18.9\\
        \hline
    \end{tabular}
\end{table}

\subsection{A mammogram screening dataset}

\begin{figure}[tb]
    \centering
    \begin{tabular}{ccc}
    \rotatebox[origin=c]{90}{
        $\widehat{Y}_{K,n+1}$  \hspace{-7.5cm}
    }
    &
    \includegraphics[width=.9\textwidth]{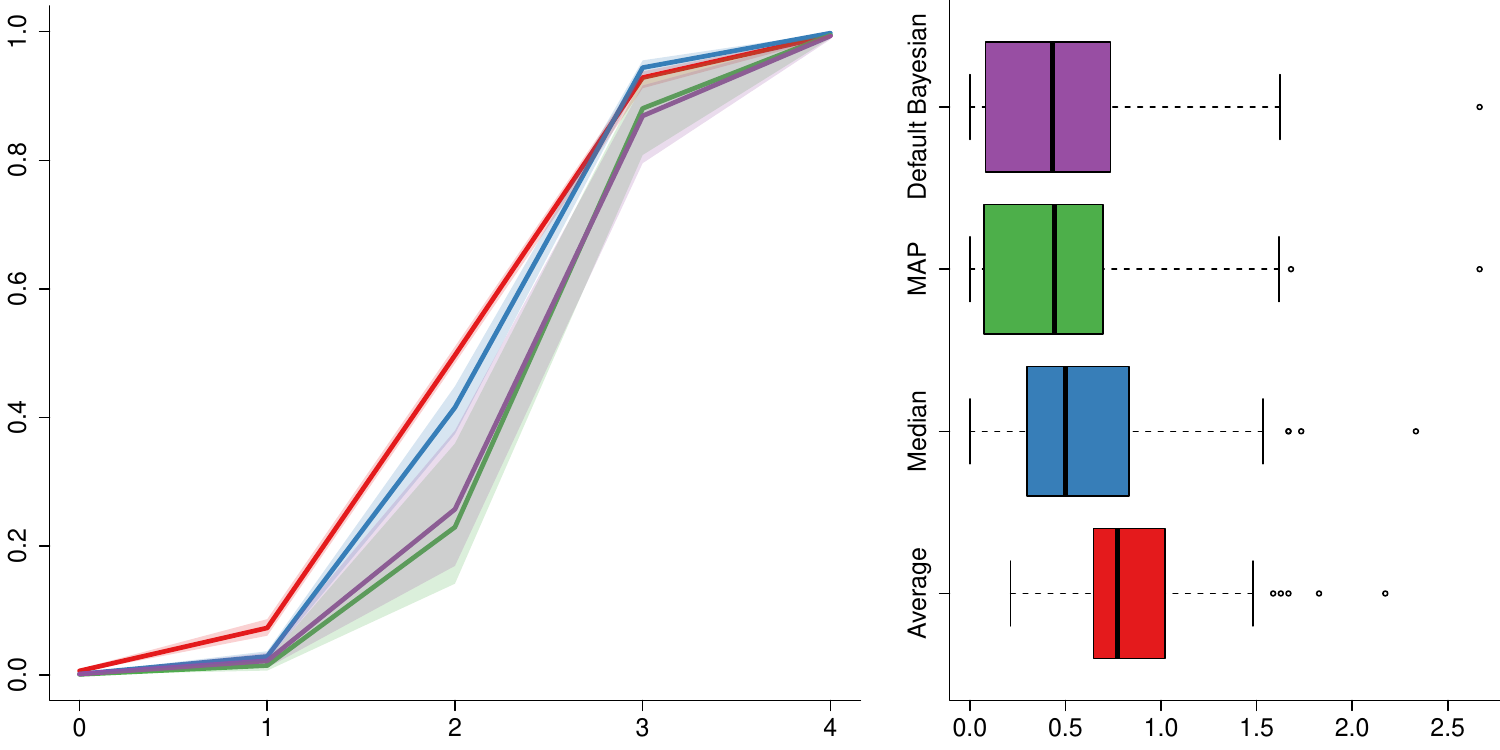}
    \\
    & \hspace{3cm}  $s_i$ \hspace{6cm} $\sum_a{\overline{\ell}_a}$-risk
    \end{tabular}
    \caption{Comparison of prediction scores $\widehat{Y}_{K, n+1}$. In the left subfigure, medians (lines) and quartiles (shaded areas) are represented. In the right subfigure, predisctions performance is measured in terms of empirical risk $\sum_a{\overline{\ell}_a}$-risk, for $a$ varying from $0.15$ to $0.45$. We run the methods on $100$ simulated mammography datasets subsampled with $4$ remaining radiologists.
    \label{fig:mammography_prediction}}
\end{figure}

In \cite{BCS03, KiLe17}, the authors consider ${\nbInd}=148$ women that may suffer or not from breast cancer. A gold-standard diagnostic test gives the infection status of the patient. 64 patients are confirmed cancer cases, and 84 are not affected. $n_{rad}=110$ radiologists read the $n$ mammography films. Further information on this dataset is given in~\ref{app:dataMammo}. Since mammography data are sensitive, they are often proprietary and confidential, generally maintained by the American College of Radiology. An agreement with the ACR is necessary to use complete datasets. We had to run simulations to synthesize complete data with only partial data available.

Moreover, error rates are nonhomogeneous in this dataset, which is out of our model. We have simulated datasets corresponding to this non-homogeneity of the error rates. The complete description of the simulation procedure is given in~\ref{app:simulmMamo}. In mammograms, authorizing a diagnosis of uncertainty is crucial, as it may correspond to situations where radiologists disagree or mammogram images are challenging to interpret. These cases correspond to situations where clinical follow-up (additional tests such as biopsies or a second reading by specialists) is preferable to a categorical decision.

We simulated $100$ datasets with a number of radiologists set to $n_i=4$. For each simulation, we used $90 \%$ of the $110$ mammograms to estimate $(\theta_{T,K}, p_K, q_K)$ for $K=A,M, \MAP$ and $B$. Next we computed the predictions $\widehat{Y}_{K, n+1}$, as described in Equations~(\ref{Eq:PredAMMAP}) and (\ref{Eq:PredB}), for the $15$ remaining mammograms. Figure~\ref{fig:mammography_prediction} displays the prediction scores versus the observed number of positive replicates $s_i=0, 1 \cdots, 4$. For each simulation, we additionally compute $\sum_a{\overline{\ell}_a}$-risk, for $a$ varying from $0.15$ to $0.45$. The default Bayesian and the Maximum-A-Posteriori methods outperform the others, whereas the average-based method provides the worst empirical risk.

\bigskip

\section{Conclusion and perspectives}
We have demonstrated the limitations of averaging binary technical replicates and highlighted the advantages of alternative methods, particularly the Bayesian approach. The median offers a simple yet more robust alternative to the mean, while the Bayesian method provides a comprehensive framework for incorporating prior knowledge and quantifying uncertainty. We also proposed a maximum penalized-likelihood approach, which can be seen as a MAP method. The MAP is a compromise between the simple median approach and the more comprehensive Bayesian framework, offering a balance between the median's computational efficiency and the Bayesian method's accuracy. These alternatives lead to more nuanced classifications, allowing for an "indecisive" outcome when the evidence is insufficient, a crucial aspect in medical diagnostics. The theoretical analysis and empirical results from both simulated and real-world medical datasets confirm the superior performance of these alternatives to the average regarding accuracy and reliability. The Bayesian method, by providing credible intervals, not only offers point estimates but also a measure of confidence, which is essential for informed decision-making in clinical practice and epidemiological studies. The predictive component of these methods provides easy-to-use guiding rules for practitioners to get a confident diagnostic based on newly observed data and the information drawn from the dataset. Indeed, a simple table indicates the predictive classification based on the numbers of replicates and positive replicates in the new data. Finally, the user-friendly R-package implementing these methods will facilitate their wider adoption and contribute to more robust and reliable analyses of binary replicate data.

This study was a first step towards theoretical results and has not incorporated any covariates so far. Future research will consider covariates to better classify the individuals, even if experts use some of them to make their decisions, instead of focusing solely on the decisions derived from these covariates. This approach would help reduce the variation among different experts, a common issue observed in the mammography screening dataset and provide mathematically proven reconciliations of independent diagnostics.

\bibliographystyle{plain}
\bibliography{filtered_references}


\appendix

\section{Technical results}\label{sec:technical}

\subsection{Full description of the maximum-a-posteriori EM estimation algorithm}\label{app:EM}

Recall the likelihood according to Equation~\ref{eq:likelihood}:
\begin{align*}
    \Lik(\theta_T, p, q|s_1,\ldots,s_n) &= \prod_{i=1}^n \left(\theta_T (1-q)^{s_i}q^{n_i-s_i} + (1-\theta_T) p^{s_i}{(1-p)}^{n_i-s_i}
    \right).
\end{align*}
This function's maximization might stick to border points, corresponding to $\hat{p}=0$ or $\hat{q}=0$, for example. Those points are not expected, and in order to avoid them, we have chosen to use priors on both $p$ and $q$ such that
\begin{align*}
    \pi(p,q)&\propto p(1-p)q(1-q),
\end{align*}
corresponding to $p,q~\sim\mathscr{B}\text{eta}(2,2)$, corresponding to weakly informative priors. The posterior distribution of the parameters is thus proportional to
\begin{align}
     \pi(\theta_T, p, q, s_1,\ldots,s_n) &=
     p(1-p)q(1-q)\prod_{i=1}^n\left(
    \theta_T (1-q)^{s_i}q^{n_i-s_i} + (1-\theta_T) p^{s_i}{(1-p)}^{n_i-s_i}
    \right).\label{eq:posteriorEM}
\end{align}
The following details the EM algorithm to maximize the problem's likelihood. Rather than maximizing the likelihood, the objective is to maximize the posterior distribution of the parameter, but the EM algorithm can be used the same way.

The EM algorithm is an iterative procedure alternating between the Estimation (E) and Maximization (M) steps. The E-step computes the expected value of the latent variables $T_i$'s given the observed data $s_i$'s and the current values of the fixed parameters. The M-step maximizes the completed 
posterior distribution
of the fixed parameters given the expected values of the latent variables. The completed
log-posterior distribution
of the fixed parameters $\theta_T$, $p$, and $q$ given the observed data $(s_1,\ldots, s_n)$ and the latent $(y_1,\ldots,y_n)$ can be extracted from \eqref{eq:Bayes1} and is, up to an additive constant,
\begin{align*}
    \pi_c(\theta_T, p, q,s_1,\ldots,s_n, y_1,\ldots, y_n)
    &= \sum_{i=1}^n \Bigg\{y_i\Big(
        s_i\log(1-q) + (n_i-s_i)\log(q) + \theta_T
    \Big) \\
    &\ +
    (1-y_i)\Big(
        s_i\log(p) + (n_i-s_i)\log(1-p) + (1-\theta_T)
    \Big) \Bigg\},\\
    &\ + \log(p(1-p)) + \log(q(1-q)).
\end{align*}
We can easily optimize the above completed log-posterior distribution 
in $\theta_T$, $p$ and $q$ given the values of $(s_1,\ldots, s_{\nbInd})$ and $(y_1,\ldots,y_{\nbInd})$. This means that the M-steps of the EM algorithm are explicit:
\begin{align}
    {\theta}_\text{T,M-step} &= \dfrac{1}{{\nbInd}}\sum_{i=1}^{\nbInd} y_i,\quad
    {p}_\text{M-step} = \dfrac{1+\sum_{i=1}^{\nbInd} s_i(1-y_i)}{2+\sum_{i=1}^{\nbInd} n_i(1-y_i)},\quad
    {q}_\text{M-step} = \dfrac{1+\sum_{i=1}^{\nbInd} (n_i-s_i)y_i}{2+\sum_{i=1}^{\nbInd} n_iy_i}.\label{equ:mlM}
\end{align}
The value of $y_i$ plugged in the M-step is the expected value of the latent variable $T_i$ given the observed data $s_i$ and the current values of the fixed parameters. (Note that those values of $y_i$ are now real numbers between $0$ and $1$ exactly as our different scores.) This is the result of the E-step in the EM algorithm. It is also explicit: since $T_i$ is a binary variable, this conditional expected value is equal to the probability of $T_i=1$ given the observed data $s_i$ and the current values of the fixed parameters. We can recognize the definition of the likelihood-based score $Y_{L,i}(\theta_T, p, q)$ given in~\eqref{eq:Y_L}. Thus, the E-step of the EM algorithm is
\begin{align}
{y_i}_\text{E-step} = Y_{L,i}(\theta_T, p, q), \label{equ:mlE}
\end{align}
where we use the last values of the fixed parameters $\theta_T$, $p$, and $q$ available during the iterative procedure.

As with many numerical optimization algorithms, the EM algorithm is sensitive to the initial values. We thus have to run the EM algorithm several times. Each time, we initialize the latent $y_i$'s by drawing them at random from
\begin{align}
y_i \sim \mathscr{B}\text{eta}(s_i + 1/2, n_i-s_i+1/2).\label{equ:mlinit}
\end{align}
We then run the EM algorithm starting with an M-step until convergence. We compare the results obtained from the several runs by using the posterior distribution 
given in~\eqref{eq:posteriorEM}, and the best model is chosen.

The posterior distribution 
of our mixture model in~\eqref{eq:posteriorEM} suffers from the label-switching problem; see Chapter 1 of~\cite{fruhwirth2019handbook}. Indeed, there is a symmetry in the posterior: we can exchange the $0$ and $1$ labels in the latent space and get the same posterior 
value with $\theta_T^\ast=1-\theta_T$, $p^\ast=1-q$ and $q^\ast=1-p$ as the new fixed parameters. However, we have assumed that the replicates are noisy measurements of the true value $T_i$. Therefore, we have set the constraint that the false-positivity rate $p$ and false-negativity rate $q$ are smaller than $1/2$. This constraint gives us a relabelling of the latent $y_i$'s that removes entirely the label-switching problem when applying the EM algorithm to our mixture model: each time the constraint on $p$ is violated, we use the symmetry described above, switch the $0$ and $1$ and change the current values of the fixed parameters to $1-\theta_T$, $1-q$ and $1-p$ and the current values of all $y_i$'s to $1-y_i$.
Algorithm~\ref{algo:em} synthesizes the complete estimation procedure.

\begin{algorithm}
\caption{Maximum-A-Posteriori estimator approximation of the fixed parameters}
\label{algo:em}
\begin{algorithmic}
\For {$r\in[\![1,R]\!]$}
    \State $\forall i\in[\![1,{\nbInd}]\!]$, initialize ${y}_\text{r,i}$: ${y}_\text{r,i} \sim \mathscr{B}\text{eta}(s_i + 1/2, n_i-s_i+1/2)$
    \While {Convergence not reached}
        \State Perform M-step:
        $${\theta}_\text{T,EM,r} \gets \dfrac{1}{{\nbInd}}\sum_{i=1}^{\nbInd} {y}_\text{r,i},\quad
    {p}_\text{EM,r} \gets \dfrac{1+\sum_{i=1}^{\nbInd} s_i(1-{y}_\text{r,i})}{2+\sum_{i=1}^{\nbInd} n_i(1-{{y}_\text{r,i}})},\quad
    {q}_\text{EM,r} \gets \dfrac{1+\sum_{i=1}^{\nbInd} (n_i-s_i){y}_\text{r,i}}{2+\sum_{i=1}^{\nbInd} n_i{y}_\text{r,i}}$$
        \If {$q>1/2$}
            \State {$(
            {\theta}_{T, \text{EM,r}},
            {p}_\text{EM,r},
            {q}_\text{EM,r})\gets (1-
            {\theta}_{T, \text{EM,r}},1-
            {p}_\text{EM,r},1-
            {q}_\text{EM,r})$}
        \EndIf
        \State $\forall i\in[\![1,n]\!]$, perform E-step:\quad
        ${y}_\text{r,i} \gets Y_{L,i}(\theta_{T,\text{EM},r}, p_{\text{EM},r}, q_{\text{EM},r})$
    \EndWhile
\EndFor
\State The best model $(\widehat{\theta}_{T, \text{MAP}},\widehat{p}_\text{MAP},\widehat{q}_\text{MAP})$ verifies
\begin{align*}
    (\widehat{\theta}_{T, \text{MAP}},\widehat{p}_\text{MAP},\widehat{q}_\text{MAP}) &=
    \argmax_{r\in[\![1,R]\!]}
    \pi({\theta}_{T, \text{EM,r}},
    {p}_\text{EM,r},
    {q}_\text{EM,r}, s_1,\ldots,s_n)
\end{align*}
\end{algorithmic}
\end{algorithm}

\subsection{Bayesian posterior calculation}\label{app:bayesian}

Using the observed value $(s_1,\ldots, s_{\nbInd})$ of the sufficient statistic $(S_1,\ldots, S_{\nbInd})$, the joint density of the Bayesian model is thus, see Figure~\ref{fig:DAG},
\begin{multline}
    \pi(\theta_T, p, q, y_1,\ldots, t_{\nbInd}, t_1, \ldots, s_{\nbInd}) = \pi(\theta_T, p, q)
    \prod_{i=1}^{\nbInd} \p_\pi\big(T_i=t_i\big|\theta_T\big)\p_\pi \big(S_i=s_i \big| T_i=t_i, p, q\big)
    \\
    = \pi(\theta_T, p, q) \prod_{i=1}^{\nbInd} \theta_T^{t_i}(1-\theta_T)^{1-t_i}
    \binom{n_i}{s_i}\Big\{
        t_i (1-q)^{s_i}q^{n_i-s_i} + (1-t_i) p^{s_i}{(1-p)}^{n_i-s_i}
        \Big\}. \label{eq:Bayes1}
\end{multline}
If we integrate over the latent $T_i$'s, we have the joint density of the observed data and the fixed parameters as
\begin{equation}
    \pi(\theta_T, p, q, s_1,\ldots,s_{\nbInd}) = \pi(\theta_T,p,q)\prod_{i=1}^{\nbInd}\binom{n_i}{s_i}\Big\{
    \theta_T (1-q)^{s_i}q^{n_i-s_i} + (1-\theta_T) p^{s_i}{(1-p)}^{n_i-s_i}
    \Big\}. \label{eq:Bayes2}
\end{equation}
Despite this explicit expression of the joint density, the posterior distribution of the fixed parameters $\theta_T$, $p$, and $q$ given the data $(s_1,\ldots, s_{\nbInd})$ is difficult to compute explicitly. We use the R-package \texttt{rstan}~\cite{Stan24}, which implements a Hamiltonian Monte Carlo algorithm to sample from the posterior distribution of the fixed parameters $\theta_T$, $p$, and $q$ given the data. Latent values of the $T_i$'s are drawn at each iteration of the MCMC algorithm from the predictive distribution of the $T_i$'s given the data $X_{ij}$'s and the current values of the fixed parameters $\theta_T$, $p$, and $q$. The predictive distribution of the $T_i$'s given the data $S_{i}$'s and the current values of the parameters $\theta_T$, $p$, and $q$ is a product of Bernoulli distributions given by~\eqref{eq:TgivenS}.

\subsection{On the likelihood-based scores}\label{sec:Lincrease}

\begin{lem}\label{lem:Lincrease}
    Assume $\theta_T\in(0,1)$, $p\in(0,1/2)$ and $q\in(0,1/2)$.
    The likelihood-based score, defined as in \eqref{eq:Y_L}, is an increasing function of $s$.
\end{lem}
\begin{proof}
    It is enough to prove that the following function is increasing in $s$:
    \[
        Y_L(s) = \theta_T{\Bigg(
            \theta_T + (1-\theta_T)\left(\frac{1-p}{q}\right)^{n-s}
            \left(\frac{p}{1-q}\right)^{s}
        \Bigg)}^{-1}.
    \]
    Actually $\displaystyle s\mapsto \left(\frac{1-p}{q}\right)^{n-s}
    \left(\frac{p}{1-q}\right)^{s}$ is a decreasing function of $s$ because
    \begin{align*}
    \left( \frac{q}{1-p} \right)^{n-s+1} \, \left( \frac{1-q}{p} \right)^{s-1}
    <  \left( \frac{q}{1-p} \right)^{n-s} \, \left(
     \frac{1-q}{p} \right)^{s}
     & \iff  \frac{q}{1-p}  <  \frac{1-q}{p} \\
     & \iff  0  <  1 - (p+q).
    \end{align*}
    And the last inequality is true because $p<1/2$ and $q<1/2$.
\end{proof}

\subsection{Proof of Lemma~\ref{lem:loss}}\label{sec:loss_proof}

Since $r(\widehat t)=\e\big[\ell(T,\widehat t)\big]$, we have
\begin{align*}
    r(0) &= \ell(0, 0) + \vartheta(\ell(1, 0) - \ell(0, 0)) = c\vartheta,
    \\
    r(1) &= \ell(0, 1) + \vartheta(\ell(1, 1) - \ell(0, 1)) = b - b \vartheta,
    \\
    r(1/2) &= \ell(0, 1/2) + \vartheta(\ell(1, 1/2) - \ell(0, 1/2))=a+(d-a)\vartheta.
\end{align*}
The three risks are thus affine functions of $\vartheta$, see Figure~\ref{fig:risks}.
\begin{figure}[tb] \centering
    \includegraphics*[width=.8\textwidth]{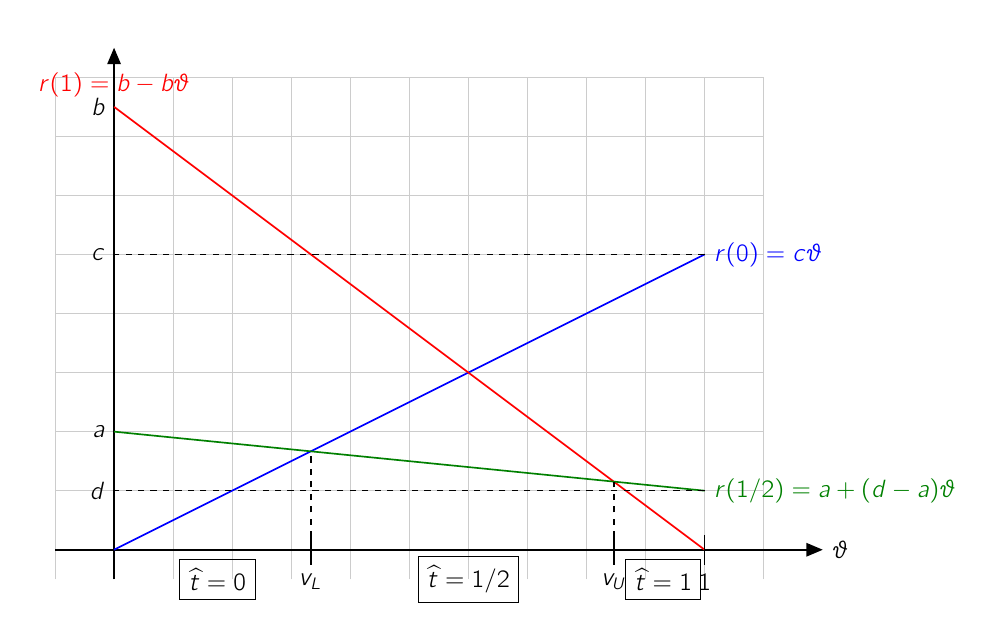}
    \caption{Risks $r\big(\widehat t\big)$ of deciding $\widehat t=0, 1/2$ or $1$ as a function of $\vartheta$. When $\vartheta\in(0, v_U)$, the lowest risk is $r(0)$, thus the best decision is $\widehat t=0$. When $\vartheta\in(v_L, v_U)$, the lowest risk is $r(1/2)$, thus the best decision is $\widehat t=1/2$. When $\vartheta\in(v_U, 1)$, the lowest risk is $r(1)$, thus the best decision is $\widehat t=1$.}\label{fig:risks}
\end{figure}

The best decision is of the desired form if and only if, as functions of $\vartheta$, the three risks $r(0)$, $r(1/2)$ and $r(1)$ intersect each other as in Figure~\ref{fig:risks}. This happens if and only if the two following conditions are satisfied:
\begin{itemize}
    \item[$(i)$] The point where the risks $r(0)$ and $r(1)$ intersect is above the line of $r(1/2)$.
    \item[$(ii)$] The slope of $r(1/2)$ is between those of $r(0)$ and $r(1)$, that is to say in $(-b, c)$.
\end{itemize}
Condition $(ii)$ is clearly equivalent to $-b < (d-a) < c$. Moreover,
the point at which $r(0)$ and $r(1)$ intersect has coordinates
\[
    \left(\frac{b}{b+c}, \frac{bc}{b+c}\right).
\]
And at $\vartheta=b/(b+c)$, the risk $r(1/2)$ is equal to $a+(d-a)b/(b+c)$. Hence, condition $(i)$ is equivalent to the first inequality in the Proposition.
Finally, the values of $v_L$ and $v_U$ are obtained by solving the equations $r(0)=r(1/2)$ and $r(1)=r(1/2)$, respectively.\qed

\subsection{Properties of the Binomial distribution}\label{sec:binomial}
\begin{figure}
    \centering
    \includegraphics[width=.7\textwidth]{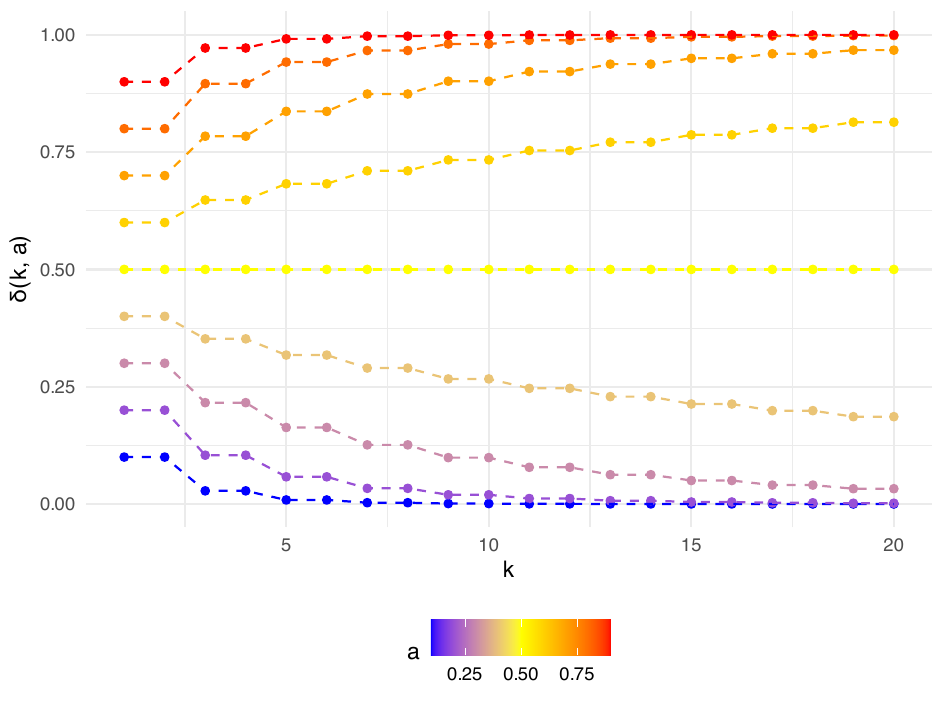}
    \caption{The values of $\delta(k,a)$ for $a=0.1, 0.2, \ldots, 0.9$, as defined in Equation~\eqref{eq:uv}.}\label{fig:delta}
\end{figure}

We need to introduce some coefficients related to the Binomial distribution to understand the probabilistic properties of the median-based scoring. For any integer $k$ and any real $a\in(0,1)$, we set
\begin{align}
    u(k, a) &= \p\Big(\mathscr{B}\text{in}(k,a)>k/2\Big),
    \quad
    v(k, a) = \p\Big(\mathscr{B}\text{in}(k,a)=k/2\Big) \text{ and}\notag
    \\
    \delta(k, a) & = u(k, a) + v(k, a)/2. \label{eq:uv}
\end{align}
Note that, If $k$ is an odd number, then $v(k, a)=0$. Values of $\delta(k, a)$ are given in Figure~\ref{fig:delta}.
\begin{prop}\label{pro:uv} The $\delta$-coefficients have the following properties.
    \begin{enumerate}[\em (i)]
        \item $\delta(1, a) = a$.\label{p1}
        \item $1-\delta(k, a) = \delta(k, 1-a)$.\label{psym}
        \item If $k$ is an even number, then $\delta(k, a) = \delta(k-1, a)$.\label{pconst}
        \item If $a<1/2$, then $\delta(k,a)$ is a non-increasing function of $k$.\label{pnoninc}
        \item If $a<1/2$, then $\delta(k,a)$ tends to $0$ when $k\to\infty$.\label{pzero}
    \end{enumerate}
\end{prop}
\begin{proof}
    We introduce $B_i$, $i=1,2,\ldots$, an iid sequence of random variables distributed according to $\mathscr{B}\text{er}(a)$, and set $Z_k=B_1+\cdots+B_k\sim \mathscr{B}\text{in}(k,a)$ for all $k\ge1$.

    \paragraph{Proof of \eqref{p1}} We have $\delta(1, a) = \p(Z_1>1/2) + \p(Z_1=1/2)/2 = \p(Z_1=1) = a$.

    \paragraph{Proof of \eqref{psym}} Consider first the case where $k$ is odd. We have $1-u(k, a) = 1 - \p(Z_k > k/2) = \p(Z_k<k/2)$ since $k/2$ is not an integer. And $\p(Z_k<k/2)= \p(k-Z_k>k/2)$. Now $k-Z_k = (1-B_1) + \cdots + (1-B_k)\sim \mathscr{B}\text{in}(k, 1-a)$. Thus, $1-u(k, a) = u(k, 1-a)$. Moreover, when $k$ is odd, $v(k, a)=0$ and $\delta(a, k) = u(a,k)$. Hence, \eqref{psym} is proven when $k$ is odd.

    Consider now the case where $k$ is even. We have $1-u(k, a) = 1 - \p(Z_k > k/2) = \p(Z_k\le k/2)= \p(Z_k < k/2) + \p(Z_k=k/2)$. Like above, $\p(Z_k<k/2)=u(k, 1-a)$. And $v(k, a)=\p(Z_k=k/2)=\p(k-Z_k=k/2)=v(k, 1-a)$. So $1-u(k, a)=\big[u(k, 1-a)+v(k, 1-a)\big]$. Thus,
    \[
        1-\delta(k,a) = 1 - u(k,a) - \frac12 v(k,a) = \big[u(k, 1-a) + v(k, 1-a)\big] - \frac12 v(k, 1-a) = \delta(k, 1-a).
    \]

    \paragraph{Proof of \eqref{pconst}} Set $k=2m$, where $m$ is an integer.
    We have $\delta(2m-1,a) = u(2m-1,a) = \p(Z_{2m-1}>m-1/2) = \p(Z_{2m-1}\ge m)$. And $Z_{2m-1} = Z_{2m} - B_{2m}$. Since $B_{2m}$ is either $0$ or $1$, we have
    \[
        \Big\{Z_{2m-1}\ge m\Big\} = \Big\{Z_{2m}\ge m+1\Big\} \cup \Big\{Z_{2m}=m, B_{2m}=0\Big\},
    \]
    and the two events are disjoint. Thus $u(2m-1,a) = u(2m,a) + \p(Z_{2m}=m, B_{2m}=0)$. To compute the last probability, recall that $Z_{2m}=B_1+\cdots+B_{2m}$, and so, $Z_{2m}=m$ means that exactly half of the $B_i$'s are equal to $0$. Thus, given $Z_{2m}=m$, the probability that $B_{2m}=0$ is $m/(2m)=1/2$. Hence $\p(Z_{2m}=m, B_{2m}=0) = \p(Z_{2m}=m)/2=v(2m,a)/2$ and finally, $u(2m-1,a) = u(2m,a) + v(2m,a)/2$.

    \paragraph{Proof of \eqref{pnoninc}} Set $a<1/2$. It is enough to prove that, for all $m$, $\delta(2m-1, a)=\delta(2m, a)>\delta(2m+1, a)$. The equality is the result of \eqref{pconst}. Thus, it is enough to prove that $\delta(2m, a)>\delta(2m+1, a)$ for all $m$.

    Fix any $m$. We have $\delta(2m+1, a) = \p(Z_{2m+1}\ge m+1)$. And $Z_{2m+1}=Z_{2m}+B_{2m+1}$ with independence between these last two variables. Thus,
    \[
        \delta(2m+1, a) = (1-a) \p(Z_{2m}\ge m+1) + a \p(Z_{2m}\ge m).
    \]
    Using $\p(Z_{2m}\ge m) = \p(Z_{2m}\ge m+1) + \p(Z_{2m}=m)$ in the above, we get
    \begin{equation*}
        \delta(2m+1, a) = \p(Z_{2m}\ge m+1) + a \p(Z_{2m}=m)=u(2m, a) + a\, v(2m,a).
    \end{equation*}
    Since $a<1/2$, we have $a\, v(2m,a)<v(2m,a)/2$ and thus $\delta(2m+1, a) < \delta(2m, a)$.

    \paragraph{Proof of \eqref{pzero}} Since $a<1/2$, Hoefding's inequality gives
    \[
        \p(Z_k\ge k/2) \le \exp\left(-2k(1/2-a)^2\right).
    \]
    Thus, $\delta(k,a) =\p(Z_k>k/2)+\p(Z_k=k/2)/2\le \p(Z_k\ge k/2)$ tends to $0$ when $k\to\infty$.
\end{proof}

\section{Theoretical results on the classifiers}\label{sec:classifiers}

\subsection{Proof of Theorem~\ref{thme:accuracy_AM_classifiers}}\label{sec:proof_accuracy_AM_classifiers}

The sensitivity of a classifier $\widehat{T}$ is $\p(\widehat{T}=1|T=1)$. The average-based classifier is one if and only if $S_i > n_i v_U$. And the median-based classifier is one if and only if $S_i > n_i/2$. Because of (H3), $\{S_i > n_i v_U \} \subset \{S_i > n_i /2 \}$, and thus the median-based classifier is more sensitive than the average-based classifier.

There is equality if and only if both events are equal. This happens if and only if $v_U=1/2$. In the case where all $n_i$ are equal, this happens if and only if $v_U < 1/2 + \delta_0$, where $\delta_0$ is defined in Section~\ref{sec:accuracy_classifiers}. The same kind of reasoning yields to the inequality on the specificities of the classifiers.\qed

\subsection{Proof of Theorem~\ref{thme:accuracy_bayesian_classifiers}}\label{sec:proof_accuracy_Bayesian_classifiers}

\paragraph{Proof of $(i)$} The likelihood-based scoring is $\p(T_i=1)$. Because of (H4) and Lemma~\ref{lem:loss}, the thresholding of this scoring with $v_L$ and $v_U$ gives the best estimator regarding $\ell$-risk. Hence, the desired inequality.\qed%

\paragraph{Proof of $(ii)$} We can also apply Lemma~\ref{lem:loss}, using the expected value $\e_\pi$ given the data $S=(S_1,\ldots, S_n)$, with $\vartheta = \p_\pi(T_i=1|S)$
which is the definition of $Y_{B,i}$. Because of (H4), Lemma~\ref{lem:loss} gives that $\widehat T_{B,i}$ is the best estimator in terms of $\ell$-posterior risk, i.e.,
\[
    \e_\pi\Big(\ell(T_i,\widehat T_{B,i})\Big|S\Big) \le \e_\pi\Big(\ell(T_i,\widehat T_i)\Big|S\Big).
\]
Integrating over the marginal distribution of $S$ in the Bayesian model yields the desired inequality.\qed%

\paragraph{Proof of $(iii)$} If $(iii)$ is wrong, then the inequality has a prior probability of $1$. I.e., with prior probability equal to $1$,
\[
    \e_\pi\Big[\ell\big(T_i,\widehat T_{i}\big)\Big|\theta_T,p,q\Big] <
    \e_\pi\Big[\ell\big(T_i,\widehat T_{B,i}\big)\Big|\theta_T,p,q\Big].
\]
If we integrate this over the prior distribution, we get a classifier $\widehat T_i$ strictly better than the Bayesian classifier $\widehat T_{B,i}$, which is impossible because of $(ii)$. Hence, $(iii)$ is true.\qed%

\subsection{Average and median-based empirical risks}
\label{sec:App_ellrisk}

Let us denote by $H$ the set of the healthy individuals and by $I$ the set of the infected ones. Let us consider a decision cost $a \in ]0, \frac{1}{2}]$. From the definition of the loss function given in subsection~\ref{sec:loss}, we have
\begin{eqnarray*}
\ell_a^M & = & a \sum_{i=1}^{\nbInd} \ind{ \frac{S_i}{n_i} = \frac{1}{2} } + \sum_{i=1}^{\nbInd} \ind{ \frac{S_i}{n_i} > \frac{1}{2} }  \ind{i \in H} + \sum_{i=1}^{\nbInd} \ind{ \frac{S_i}{n_i} < \frac{1}{2} }  \ind{i \in I}.
\end{eqnarray*}
Then $\ell_a^M$ is a linear and increasing function of the decision cost $a$. In a similar way,
\begin{eqnarray*}
\ell_a^A & = & a \sum_{i=1}^{\nbInd} \ind{ \frac{S_i}{n_i} \in [a, 1-a] } + \sum_{i=1}^{\nbInd} \ind{ \frac{S_i}{n_i} \in ]1-a, 1] }  \ind{i \in H} + \sum_{i=1}^{\nbInd} \ind{ \frac{S_i}{n_i} \in [0, a[ }  \ind{i \in I}.
\end{eqnarray*}
It is obvious that as long as $a$ (or $1-a$) does not cross an observable value of $Si/ni$, the terms involved in the sums are constant, and so the average-risk $\ell_a^A$ is linear and increasing with $a$. Then jumps occur
only when $a$ crosses an observed value of $\frac{s_i}{n_i}$. In this case, the average risk increases from $1-a$ times the number of infected patients with $\frac{s_i}{n_i} = a$ but decreases from $a$ times the number of healthy patients with $\frac{s_i}{n_i} = a$. At the same time, the average risk increases from $1-a$ times the number of healthy patients with $\frac{s_i}{n_i} = 1-a$ but decreases from $a$ times the number of infected patients with $\frac{s_i}{n_i} = 1-a$. Note that infected patients are mainly expected to have $\frac{s_i}{n_i} > \frac{1}{2}$ whereas healthy patients are mainly expected to have $\frac{s_i}{n_i} < \frac{1}{2}$.

\section{Theoretical results on the scores}\label{sec:scoring}

\subsection{Bias of the scores}\label{sec:bias_scores}

For an individual $i$, the two frequentist scores $Y_{A, i}$ and $Y_{M, i}$ are biased when we consider recovering the latent $T_i$ variable or the prevalence $\theta_T$. Their expected and conditional expected values given $T_i$ are explicit; see Proposition~\ref{pro:bias} below.
\begin{prop}[Bias]\label{pro:bias}
    Fix an individual $i$ for which we observe $n_i$ replicates. The average-based scoring $Y_{A,i}$ is such that
    \begin{align*}
        \e\big(Y_{A,i}\big|T_i\big) & = T_i (1-q) + (1-T_i)p,\\
        \e\big(Y_{A,i}\big) & = \theta_T (1-q) + (1-\theta_T)p.
    \end{align*}
    The median-based scoring $Y_{M,i}$ is such that
    \begin{align*}
        \e\big(Y_{M,i}\big|T_i\big) & = T_i \delta(n_i,1-q)+ (1-T_i) \delta(n_i,p),
        \\
        \e\big(Y_{M,i}\big) & = \theta_T \delta(n_i,1-q)
        + (1-\theta_T) \delta(n_i,p),
    \end{align*}
    with the coefficients defined in Equation~\eqref{eq:uv}.
\end{prop}
\begin{proof} For all $K\in\{A, M\}$, $\e(T_{K,i}|T_i)$ is a linear function of $T_i$. Thus, computing $\e(T_{K,i})$ from $\e(T_{K,i}|T_i)$ is straightforward: we replace $T_i$ in the expression by its expected value $\theta_T$.

The two scores are explicit functions of $S_i$, given by $Y_{A, i} = S_i/n_i$  and Equation~\eqref{eq:Y_M1} above. Integrating those formulae over the Binomial distribution of $[S_i|T_i]$ in~\eqref{eq:SgivenT} yields the desired results.
\end{proof}

\subsection{Variance of the scores}\label{sec:variance_scores}
The variance measures the random variability of each score around its expected value. We can compute them explicitly for the average- and median-based scores as follows. To state the results, we need to introduce another coefficient as in Equation~\eqref{eq:uv}:
\begin{equation}\label{eq:gamma}
    \gamma(k, a) = u(k, a)(1-u(k, a)) + \frac14 v(k, a)(1-v(k, a))
    -  u(k, a)v(k, a).
\end{equation}
\begin{prop}[Variance]\label{pro:var}
    Fix an individual $i$ for which we observe $n_i$ replicates. The average-based scoring $Y_{A,i}$ is such that
    \begin{align*}
        \var\big(Y_{A,i}\big|T_i\big) & = T_i \frac{q(1-q)}{n_i} + (1-T_i)\frac{p(1-p)}{n_i},\\
        \var\big(Y_{A,i}\big) & = \theta_T(1-\theta_T)(1-q-p)^2+\theta_T\frac{q(1-q)}{n_i} + (1-\theta_T)\frac{p(1-p)}{n_i}.
    \end{align*}
    The median-based scoring $Y_{M,i}$ is such that
    \begin{align*}
        \var\big(Y_{M,i}\big|T_i\big) & = T_i\gamma(n_i, 1-q)+ (1-T_i)\gamma(n_i, p),
        \\
        \var\big(Y_{M,i}\big) & = \theta_T(1-\theta_T)\big(1-\delta(n_i, q)-\delta(n_i, p)\big)^2 +\theta_T\gamma(n_i, 1-q) + (1-\theta_T) \gamma(n_i, p).
    \end{align*}
\end{prop}
\begin{proof}
    For the average-based score $Y_{A,i}$, we used $Y_{A,i}=S_i/n_i$ and the Binomial distribution of $[S_i|T_i]$ in~\eqref{eq:SgivenT} to compute the conditional variance given $T_i$. The unconditional variance is obtained by using
    \[
        \var(Y_{A,i}) = \e(\var(Y_{A,i}|T_i)) + \var(\e(Y_{A,i}|T_i)).
    \]

    Likewise, for the variance of the median-based scoring $Y_{M,i}$, using Equation~\eqref{eq:Y_M1}.
\end{proof}

\section{Theoretical results on the prevalence estimates}\label{sec:prevalence}

\subsection{Prevalence estimates bias}\label{sec:prevalence_bias}
There is a systematic bias in the average-based and median-based prevalence estimates, which are empirical means of their respective scores. Indeed, using Proposition~\ref{pro:bias}, we can compute their bias. We need the following coefficients to state the results: for all $a\in(0,1)$, we set
\[
    \ol\delta(a) = \frac{1}{n}\sum_{i=1}^n\delta(n_i, a),
\]
with the coefficients introduced in Equation~\eqref{eq:uv}.
\begin{prop}\label{prop:prevalence_bias}
    The average-based prevalence estimate $\widehat{\theta}_{T,A}$ is such that
    \[
        \e(\widehat{\theta}_{T,A}) = \theta_T (1-q) + (1-\theta_T)p.
    \]
    The median-based prevalence estimate $\widehat{\theta}_{T,M}$ is such that
    \[
        \e(\widehat{\theta}_{T,M}) = \theta_T (1-\ol\delta(q)) + (1-\theta_T)\ol\delta(p),
    \]
    Moreover, if {\em(H\ref{h3replicates})} is true, then
    \[
        \ol\delta(p) < p \quad \text{and} \quad \ol\delta(q) < q.
    \]
\end{prop}

\subsection{Proof of Theorem~\ref{thme:bias_AM}}\label{sec:proof_prevalence}

\begin{figure}
    \centering
    \includegraphics*[width=.7\textwidth]{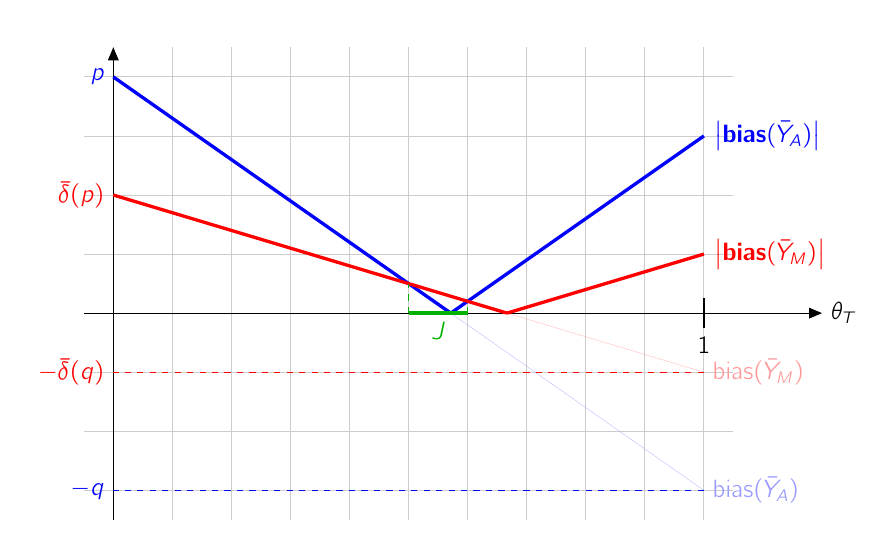}
    \caption{Typical behavior of the Biases of the average- and median-based prevalence estimates $\widehat{\theta}_{T,A}$ and $\widehat{\theta}_{T,M}$ as a function of $\theta_T$, and their absolute values. In absolute values, the bias of the average-based estimate in thick blue is smaller than the bias of the median-based one if and only if $\theta_T\in J$.}\label{fig:bias}
\end{figure}

Fix $p$ and $q$ in $(0,1/2)$ and assume (H2). The bias of the prevalence estimates are
the linear functions of $\theta_T$ because of Proposition~\ref{prop:prevalence_bias}:
\begin{align*}
    \text{bias}(\widehat{\theta}_{T,A}) &= \e(\widehat{\theta}_{T,A}) - \theta_T  = p - \theta_T(p+q),\\
    \text{bias}(\widehat{\theta}_{T,M}) &= \e(\widehat{\theta}_{T,M}) - \theta_T  = \ol\delta(p) - \theta_T (\ol\delta(p)+\ol\delta(q)).
\end{align*}
The first one is null at $\theta_T=p/(p+q)$, and the second one is null at $\theta_T=\ol\delta(p)/(\ol\delta(p)+\ol\delta(q))$. The absolute values of these biases are thus as in Figure~\ref{fig:bias}. And, because of (H2), $\ol\delta(p)<p$ and $\ol\delta(q)<q$. Thus,
\[
    \Big|\text{bias}(\widehat{\theta}_{T,M})\Big| < \Big|\text{bias}(\widehat{\theta}_{T,A})\Big|
\]
except when $\theta_T$ is in the interval, we denote $J$, where the blue curve is below the red one in Figure~\ref{fig:bias}. And, since the blue curve is null at $p/(p+q)$, this value is in $J$.

Because of Proposition~\ref{pro:uv}, $\delta(k,p)$ is a non-increasing function of $k$, and thus $\ol\delta(p)\le\delta(n_0, p)$. And likewise, $\ol\delta(q)\le\delta(n_0, q)$. Moreover, the two bounds $\delta(n_0,p)$ and $\delta(n_0,q)$ tends to $0$ when $n_0\to\infty$. That is to say that the red curve in Figure~\ref{fig:bias} tends to zero when $n_0\to\infty$, whereas the blue curve does not change. Hence, the length of $J$ tends to $0$ when $n_0\to\infty$.\qed%

\subsection{Variance of the prevalence estimates}\label{sec:prevalence_variance}

The variance of the average- or median-based estimate is the variance of the empirical means of the scores. As for the biases, we can compute them explicitly. We need the following coefficients to state the results: for all $a\in(0,1)$, we set
\[
    \ol\gamma(a) = \frac{1}{{\nbInd}}\sum_{i=1}^{\nbInd}\gamma(n_i, a)\quad \text{and} \quad
    \tilde n = \left(\frac{1}{{\nbInd}}\sum_{i=1}^{\nbInd} \frac{1}{n_i}\right)^{-1}.
\]
Note that $\tilde n$ is the harmonic mean of the $n_i$'s.
\begin{prop}\label{prop:prevalence_variance}
    The average-based prevalence estimate $\widehat{\theta}_{T,A}$ is such that
    \[
        \var(\widehat{\theta}_{T,A}) = \frac1{\nbInd} \Bigg\{
            \theta_T(1-\theta_T)(1-q+p) + \theta_T\frac{q(1-q)}{\tilde n} + (1-\theta_T)\frac{p(1-p)}{\tilde n}
        \Bigg\}.
    \]
    The median-based prevalence estimate $\widehat{\theta}_{T,M}$ is such that
    \[
        \var(\widehat{\theta}_{T,M}) = \frac1{\nbInd} \Bigg\{
            \theta_T(1-\theta_T)(1-\ol\delta(q)+\ol\delta(p)) + \theta_T\ol\gamma(1-q) + (1-\theta_T)\ol\gamma(p)
        \Bigg\}.
    \]
\end{prop}
\begin{proof}
    Both prevalence estimates are empirical means of their independent, respective scores. Thus, the proposed formulae are straightforward consequences of Proposition~\ref{pro:var}, which gives the variances of the scores.
\end{proof}

\subsection{Proof of Theorem~\ref{thme:accuracy_Bayesian_prevalence}}\label{sec:proof_accuracy_Bayesian_prevalence}

The proof is similar to that of Theorem~\ref{thme:accuracy_bayesian_classifiers} given in Section~\ref{sec:proof_accuracy_Bayesian_classifiers}. We have to replace the loss function $\ell$ by the squared error loss function $L^2(\theta_T, \widehat\theta) = (\theta_T-\widehat\theta)^2$. And, to replace Lemma~\ref{lem:loss}, we use the fact that the best estimator in terms of the posterior $L^2$-risk is the posterior mean. Likewise, it achieves the minimum Bayesian $L^2$-risk and is admissible.\qed%

\section{Real case data information}

\subsection{Details of the mammogram dataset}\label{app:dataMammo}

Three radiologists did not consider all the exams: radiologist code \texttt{1201} to whom 63 positive cases are presented, corresponding to 1 missing positive case exam. Radiologist \texttt{7714} to whom 61 positive cases and 79 negative cases are presented, corresponding to 3 missing positive case exams and five missing negative case exams. Radiologist \texttt{9007} to whom 83 negative cases are presented, corresponding to one missing negative case exam.
Thus, four missing positive case exams and six missing negative case exams among the $n_{rad}$ radiologists.

For any radiologist $j$, it is possible to estimate his false-positivity rate $p_j$, and his false-negativity rate $q_j$ as
\begin{align}
    \widehat{p}_j = \dfrac{
    \sum_{i=1}^{\nbInd} \ind{T_{i}=0}\ind{X_{ij}=1}\ind{m_{ij}=1}
    }{
    \sum_{i=1}^{\nbInd} \ind{T_{i}=0}\ind{m_{ij}=1}
    }
    \quad\text{and}\quad
    \widehat{q}_j =  \dfrac{
    \sum_{i=1}^{\nbInd} \ind{T_{i}=1}\ind{X_{ij}=0}\ind{m_{ij}=1}
    }{
    \sum_{i=1}^{\nbInd} \ind{T_{i}=1}\ind{m_{ij}=1}
    },\label{equ:simupq}
\end{align}
where $m_{ij}$ is equal to $1$ if radiologist $j$ observes mammography $i$ and $0$ otherwise.

\subsection{A simulation procedure for synthetic mammography datasets}\label{app:simulmMamo}

This paragraph details a simulation procedure to generate a dataset $\boldsymbol{X}=\left\{\left\{x_{ij}\right\}_{j},t_i,n_i\right\}_{i}$. $x_{ij}$s are filled with 0 and 1 corresponding to realizations of the variables $X_{ij}$ for each of the $n$ patients and each of the $n_{rad}$ radiologists. This dataset must verify different points:
\begin{itemize}
    \item $t_i$ is set equal to 0 for $i\in[\![1,84]\!]$ and to 1 for $i\in[\![85,148]\!]$
    \item according to the equations~\eqref{equ:simupq} and for any radiologist $j$, the $\left(x_{ij}\right)_{i}$s are such as
    \begin{itemize}
        \item there are exactly $84\widehat{p}_j$ 1s and $84(1-\widehat{p}_j)$ 0s in the 84 first coordinates of $\left(x_{ij}\right)_{i}$
        \item there are exactly $64\widehat{q}_j$ 0s and $64(1-\widehat{q}_j)$ 1s in the 64 last coordinates of $\left(x_{ij}\right)_{i}$
    \end{itemize}
    \item $n_i$ is set equal to $n_{rad}$ except for radiologists of ids 1201, 7714 and 9007, as discussed in the following point.
    \item concerning the missing values discussed earlier:
    \begin{itemize}
        \item there is 1 missing positive case for radiologist of id 1201 and $n_i$ is set to 109
        \item there are 3 missing positive cases and 5 missing negative cases for radiologist of id 7714  and $n_i$ is set to 102
        \item there is 1 missing negative case for radiologist of id 9007 and $n_i$ is set to 109.
    \end{itemize}
    \item among the patients, some must be hard to be properly diagnosed.
\end{itemize}
For the last point we have chosen the following procedure for any radiologist $j$: treat $x_j^{(0)}:=\left\{x_{ij}\right\}_{i\in[\![1,84]\!]}$ and $x_j^{(1)}:=\left\{x_{ij}\right\}_{i=[\![85,148]\!]}$ in parallel. In $x_j^{(0)}$, there are exactly $84\widehat{p}_j$ values equal to 1 and $84(1-\widehat{p}_j)$ values equal to 0. In $x_j^{(1)}$, there are exactly $64\widehat{q}_j$ values equal to 0 and $64(1-\widehat{q}_j)$ values equal to 1. In both cases, we throw the desired positions of $84\widehat{p}_j$ values among the vector $[\![1,84]\!]$, respectively the desired positions of $64\widehat{q}_j$ values among the vector $[\![85,148]\!]$. Each patient is weighted, based on weights $w_i$, corresponding to the difficulty of correctly diagnosing patient $i$ (say \textit{$x_{ij}=1$ while $t_i=0$} and \textit{$x_{ij}=0$ while $t_i=1$}) meaning that if $w_k>w_i$ then patient $i$ is harder to diagnose correctly than patient $k$.
Note, there are many possible datasets:
\begin{align*}
    \prod_{j=1}^{n_{rad}}
    C_{64}^{64\widehat{q}_j}
    C_{84}^{84\widehat{p}_j},
\end{align*}
in the mean case $\widehat{p}_j$ are all equal to $1/N\sum_{j=1}^N\widehat{p}_j\approx0.22$ and $\widehat{q}_j$ are all equal to $1/N\sum_{j=1}^N\widehat{q}_j\approx0.13$ and the formula simplifies in $\left(C_{64}^{32}C_{84}^{42}\right)^{110}$, which is not reachable by common computers. So, testing all the possible datasets might be inappropriate. Although it is possible to derive the law of $x_j^{(0)}$ and $x_j^{(1)}$. Let us denote
\begin{align*}
    \mathcal{U}([\![\alpha,\beta]\!],n_0)=\left\{
        \text{Draw of size $n_0$ in $[\![\alpha,\beta]\!]$ without replacement}
    \right\}.
\end{align*}
Then $u\in \mathcal{U}([\![1,84]\!],84\widehat{p}_j)$ gives the indices of the non-diseased women who are diagnosed as diseased by the radiologist. Denoting $u=\{u_{i}\}_{i\in[\![1,84\widehat{p}_j]\!]}$:
\begin{align*}
    \p(\{x_{ij}^{(0)}\}_{i\in u}=1,\{x_{ij}^{(0)}\}_{i\in u^c}=0|\widehat{p}_j) =
    \dfrac{
        \prod_{i\in [\![1,84\widehat{p}_j]\!]}w_{u_i}
    }{
        \prod_{i\in  [\![1,84\widehat{p}_j]\!]}\left(
            \sum_{k\in u_{[\![i,84\widehat{p}_j]\!]}}w_k
            +
            \sum_{k\in u^c}w_k
        \right)
    },
\end{align*}
where $u^c$ is the complementary of $u$ in $[\![1,84]\!]$. Equivalently, for any $v\in\mathcal{U}([\![85,148]\!],64\widehat{q}_j)$:
\begin{align*}
    \p(\{x_{ij}^{(1)}\}_{i\in v}=0,\{x_{ij}^{(1)}\}_{i\in v^c}=1|\widehat{q}_j) =
    \dfrac{
        \prod_{i\in [\![85,64\widehat{q}_j]\!]}w_{v_i}
    }{
        \prod_{i\in [\![85,64\widehat{q}_j]\!]}\left(
            \sum_{k\in v_{[\![i,64\widehat{q}_j]\!]}}w_k
            +
            \sum_{k\in v^c}w_k
        \right)
    },
\end{align*}
where $v^c$ is the complementary of $v$ in $[\![85,148]\!]$.
These probabilities were calculated by considering the order, which is the solution chosen by the function \texttt{sample} of the software \textbf{R}~\cite{rsoftware}.

\begin{figure}
    \centering
    \includegraphics[width=\textwidth]{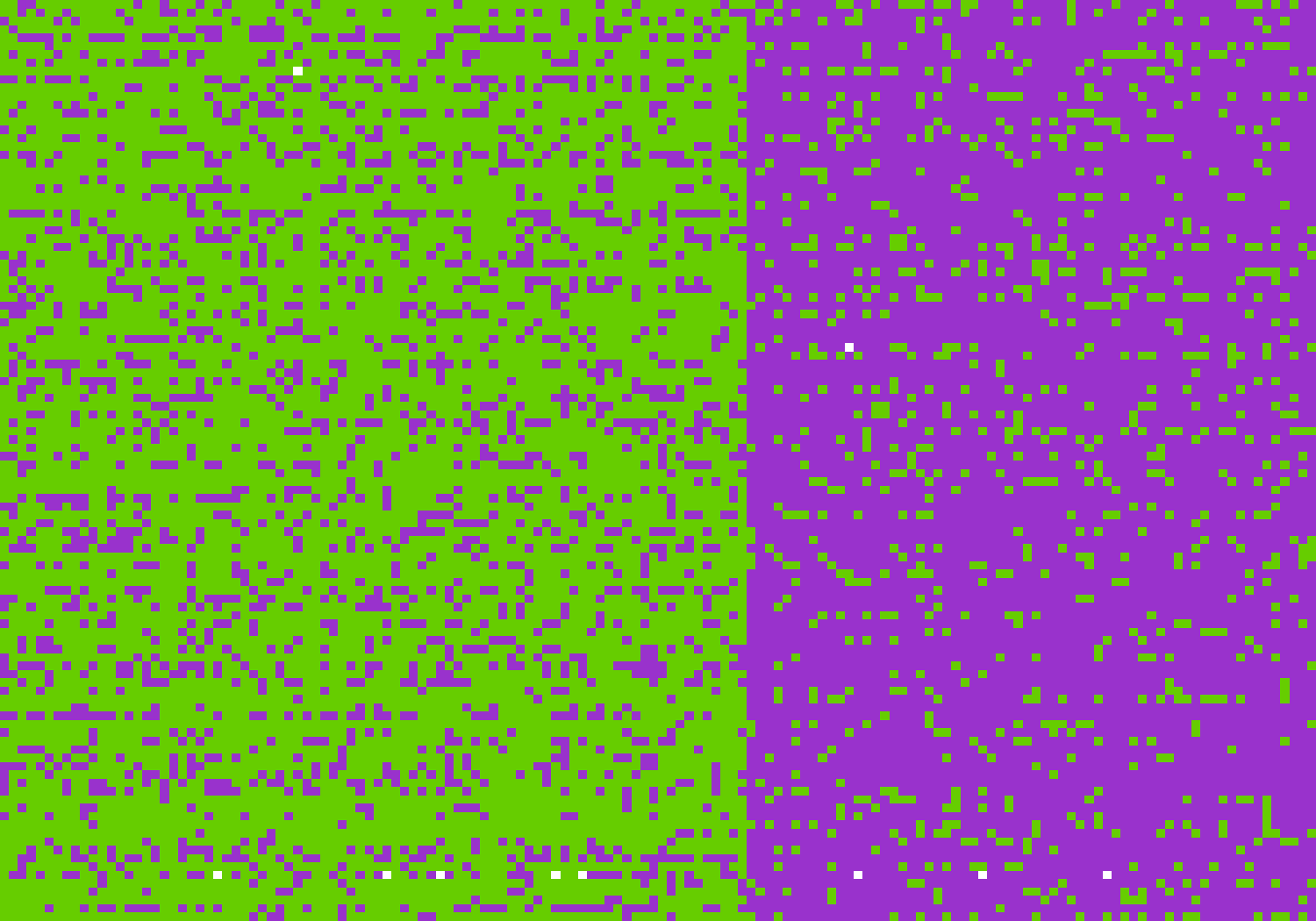}
    \caption{Example of a simulated dataset for the mammography analysis where green pixels correspond to 0's, purple pixels correspond to 1's, and white corresponds to NA's. Each row is a radiologist, and each column is a mammogram.}
    \label{fig:mamographyExampleSimuData}
\end{figure}

\end{document}